%% file: main.tex
\documentclass{article}
\usepackage{graphicx}
\input{preamble}

\newcommand{\rvec}{\vec{r}}

\newcommand{\xvec}{\vec{y}}
\newcommand{\y}{y}
\newcommand{\zvec}{\vec{z}}
\newcommand{\ET}{ET}
\newcommand{\bs}{\mathbf{s}}
\usepackage[T1]{fontenc}
\usepackage{enumitem}
\title{A New Quantum Linear System Algorithm Beyond the Condition Number and Its Application to Solving Multivariate Polynomial Systems}
\author{
Jianqiang Li \thanks{Department of Computer Science, The University of Texas at Dallas, TX  {\tt Jianqiang.Li@UTDallas.edu}}
\thanks{Department of Computer Science, Rice University, Houston, TX}
}
\begin{document}
\maketitle

\begin{abstract}
Given a matrix $A_{M \times N}$ and a vector $\vec{b}$, the quantum linear system (QLS) problem asks for the preparation of a quantum state $\ket{\vec{y}}$ proportional to the solution of $A\vec{y} = \vec{b}$. Existing QLS algorithms typically have runtimes that scale linearly with the condition number $\kappa(A)$, the sparsity of $A$, and logarithmically with the inverse precision. However, these algorithms often overlook structural properties of the input vector $\vec{b}$, despite the fact that its alignment with the eigenspaces of $A$ can significantly affect performance.

In this work, we present a new QLS algorithm that explicitly leverages the structure of the right-hand side vector $\vec{b}$. The runtime of our algorithm depends polynomially on the sparsity $\s$ of the augmented matrix $H = [A,  -\vec{b}]$, the inverse precision, the norm of the solution $\vec{y} = A^+ \vec{b}$, and a new instance-dependent parameter
$$
ET = \sum_{i=1}^M p_i^2 \cdot d_i,
$$
where $\vec{p} = (AA^{\top})^+ \vec{b}$, and $d_i$ denote the squared $\ell_2$-norm of the $i$-th row of $H$. To further reduce the runtime, we introduce a structure-aware rescaling technique tailored to the solution $\vec{y} = A^+ \vec{b}$.

This combination of an instance-aware QLS algorithm and a rescaling strategy opens new avenues for achieving superpolynomial quantum speedups in various domains. As a concrete application, we develop a quantum algorithm for solving multivariate polynomial systems in regimes where previous QLS-based methods fail. In particular, we apply our approach to the maximum independent set (MIS) problem, formulated as a special case of a polynomial system. Given a graph, the MIS problem asks for finding the largest subset of vertices such that no two vertices in the subset share an edge. We provide a detailed runtime analysis and show that, under certain conditions, our quantum algorithm for the MIS problem runs in polynomial time while no efficient classical algorithm is currently known.

\end{abstract}

\section{Introduction}

Quantum linear system (QLS) algorithms solve $A\xvec = \bvec$ by outputting a quantum state $\ket{\xvec}$ that encodes the solution, achieving polylogarithmic dependence on the system dimension, whereas classical algorithms require polynomial time~\cite{HHL09,childs2015QLinSysExpPrec,ambainis2010VTAA,gilyen2018QSingValTransf,subacsi2019quantum,lin2020optimal,costa2022optimal,dalzell2024shortcut,morales2024quantum,low2024quantum}. The QLS problem is BQP-complete~\cite{HHL09}, underscoring its potential as a universal subroutine for quantum speedups. However, any end-to-end application must simultaneously satisfy four requirements: efficient input preparation, sparse matrix access, a manageable condition number, and efficient output extraction~\cite{aaronson2015ReadTheFinePrint}. Meeting all four at once has remained elusive, and no ``killer'' application of QLS algorithms has been identified to date.

\paragraph{Why existing approaches fall short.}
A key reason for the absence of a ``killer'' application, we argue, is that existing QLS algorithms treat $A$ and $\bvec$ as independent inputs. Their complexity depends on the spectral properties of $A$ alone, most prominently the condition number $\kappa(A)$, and entirely ignores how $\bvec$ aligns with the eigenspaces of $A$. When $\bvec$ is concentrated in certain subspaces of the column space of $A$, the problem can be substantially easier, yet no existing algorithm fully exploits this. Three lines of work have attempted to incorporate $\bvec$, each facing fundamental barriers:

\begin{itemize}
\item \emph{Truncated QLS}~\cite{HHL09}: inverts only the well-conditioned component of $\bvec$, discarding its projection onto ill-conditioned subspaces. For general $\bvec$, choosing the truncation threshold optimally without incurring the full condition number dependence remains open.

\item \emph{Preconditioned QLS}~\cite{tong2021fast}: achieves complexity linear in $\|A^{+}\bvec\|$, but applies only to a restricted family of fast invertible linear systems.

\item \emph{Multidimensional electrical networks}~\cite{li2025multidimensional}: 
for graph-related linear systems $B\xvec = \vec{b}$~\cite{wang2017efficient,spielman2019spectral,vishnoi2013lx}, 
achieves complexity depending on $\|\pvec\| = \|(BB^{\top})^{+}\bvec\|$, 
unifying quantum walk algorithms derived from electrical 
networks~\cite{belovs2013ElectricWalks,apers2022elfs,piddock2019electricfind,jeffery2023multidimensional}. 
This result genuinely incorporates $\bvec$ into the complexity, 
but the framework relies on the graph structure
and does not extend to general sparse linear systems.
\end{itemize}

\paragraph{Our approach: instance-aware QLS.}
We propose a new paradigm for QLS that analyzes $A$ and $\bvec$ jointly. Generalizing the multidimensional electrical network framework of~\cite{li2025multidimensional} from graph-related systems to arbitrary linear systems, we obtain a QLS algorithm (\cref{alg:newQLS2}) whose complexity replaces the condition number $\kappa(A)$ with an instance-dependent parameter
\[
  \ET \;=\; \sum_{i=1}^{M} p_i^2 \cdot d_i,
\]
where $\pvec = (AA^{\top})^{+}\bvec$ and $d_i$ is the squared $\ell_2$-norm of the $i$-th row of the augmented matrix $H = [A, -\bvec]$. The parameter $\ET$ quantifies how favorably $\bvec$ aligns with the eigenspaces of $A$: when the alignment is favorable, $\ET$ is small even if $\kappa(A)$ is large; in the worst case, $\ET \leq \mathrm{poly}(\kappa(A))$ under $\|A\|\le 1, \|\bvec\|=1$. Our algorithm achieves complexity\footnote{We use $\widetilde{O}(\cdot)$ to suppress $\mathrm{polylog}(M \cdot N)$ factors.}
\[
  \widetilde{O}\!\Bigl(
     \|\xvec\|^{2}\,
    \frac{\sqrt{s^3 \cdot \ET}}{\epsilon^{2}}
  \Bigr),
\]
where $\xvec = A^{+}\bvec$ with $\|\xvec\|\geq 1$ and $s$ is the sparsity of $H$ (the maximum number of nonzero entries in any row or column). When $\|\xvec\|$ is known and the system is normalized so that $\|\xvec\| = \Theta(1)$, the complexity simplifies to $\widetilde{O}(\sqrt{s^3 \cdot \ET}/\epsilon^2)$.

Table~\ref{tab:qls} summarizes the comparison between our instance-aware QLS algorithm and existing QLS algorithms.

\begin{table}[ht]
\centering
\renewcommand{\arraystretch}{1.4}
\small
\begin{tabular}{|c|c|c|}
\hline
 & \textbf{Existing QLS algorithms}
 & \textbf{New QLS algorithm (\cref{alg:newQLS2})} \\
\hline
Input
  & Prepare $\ket{\bvec}$
  & Prepare $\ket{\Phi_b}$ (column star state; see \cref{def:colstates}) \\
\hline
Sparsity
  & Sparse matrix $A$
  & Sparse matrix $H=[A,-\bvec]$ \\
\hline
Key parameter
  & $\kappa(A)$
  & $\ET = \sum_{i=1}^M p_i^2 \cdot d_i$ \\
\hline
Complexity
  & $\widetilde{O}(s \cdot \kappa(A) \cdot
    \log(1/\epsilon))$
  & $\widetilde{O}\!\bigl((\|\xvec\|_2^{2}\sqrt{s^3 \cdot \ET}
    /\epsilon^2\bigr)$ \\
\hline
\end{tabular}
\caption{Comparison of QLS algorithms. Here $\pvec = (AA^{\!\top})^{+}\bvec$, $d_i = \|H_{i,*}\|_2^2$, and $s$ is the sparsity of $H = [A,-\bvec]$. When $\|\xvec\|$ is known and normalized to $\Theta(1)$, the complexity simplifies to $\widetilde{O}(\sqrt{s^3 \cdot \ET}/\epsilon^2)$.}
\label{tab:qls}
\end{table}

To illustrate that $\ET$ can be dramatically smaller than $\kappa(A)$, 
we present two separation witnesses (see Section~\ref{sec:examples} for details). 
First, for the diagonal system 
$A = \operatorname{diag}(\mathbf{I}_{n-1}, 1/n)$ with 
$\bvec = (0,\dots,0,1/n)^{\top}$, 
one has $\kappa(A) = n$ yet $\ET = O(1)$, 
so Algorithm~\ref{alg:newQLS2} runs in time $\widetilde{O}(1/\epsilon^2)$ 
independent of $n$. 
Notably, truncated-HHL variants must set their inversion threshold 
$\tau \le 1/n$ to recover the correct solution, 
offering no improvement over standard HHL on this instance. 
Second, for a linear system derived from the welded-tree graph 
of depth $n$, the condition number is $\kappa(A) = \Omega(2^{n/2})$ 
while $\ET = \operatorname{poly}(n)$.

\paragraph{QLS algorithms for polynomial systems.}
A natural testbed for our instance-aware QLS algorithm is solving multivariate polynomial systems $\{f_1(x_1,\ldots,x_n)=0,\ldots, f_m(x_1,\ldots,x_n)=0\}$~\cite{chen2022quantum,ding2023limitations}. This problem reduces to solving the Boolean Macaulay linear system $A\xvec = \bvec$, an exponentially large but polynomially sparse linear system whose columns are indexed by all multilinear monomials in $n$ Boolean variables~\cite{ding2023limitations}. Three of the four caveats, input preparation, sparsity, and output extraction, are already satisfied in this setting; the condition number is the sole remaining bottleneck.

\paragraph{Barrier: exponential truncated QLS condition number.}
Our instance-aware QLS algorithm does not resolve this bottleneck on its own. The solution satisfies $\|\xvec\|_2 = \sqrt{2^h - 1}$, where $h$ is the Hamming weight of the Boolean solution, which forces $\|\xvec\|_2 = \Omega(2^{h/2})$. This is not an artifact of our parameter; \cite{ding2023limitations} established that the truncated QLS condition number $\kappa_{\bvec}(A) = \|A^{+}\bvec\| = \|\xvec\|_2$ is a tight lower bound for \emph{any} QLS algorithm that preserves the solution, including all known preconditioning methods~\cite{clader2013preconditioned,shao2018quantum}. As long as the solution vector remains unchanged, no QLS algorithm, whether based on the condition number or on $\ET$, can solve these systems in polynomial time when $h$ is large.

\paragraph{Right-rescaling to reduce solution complexity.}
The idea of modifying a linear system to reduce solution norm, at the cost of changing the solution vector, already appears in the oracle setting: \cite{li2025multidimensional} introduced additional structural constraints to demonstrate an exponential quantum--classical separation for pathfinding on the welded tree circuit graph. We develop this idea in a different direction. Rather than left-multiplying the system as in standard preconditioning ($DA\xvec = D\bvec$), which preserves the solution and therefore preserves the lower bound, or adding graph-structural constraints as in~\cite{li2025multidimensional}, which requires oracle access, we \emph{right}-multiply: solve $AD\zvec = \bvec$, where $D$ is a diagonal matrix satisfying $\xvec = D\zvec$.

By choosing $D$ appropriately, we reduce the solution norm from $\|\xvec\|_2 = \sqrt{2^h - 1}$ to $\|\zvec\|_2 = \sqrt{h}$ and, more importantly, reduce $\ET$ to a potentially polynomial quantity. The condition number $\kappa(AD)$ remains exponentially large (\cref{prop:weighted_condition_number}), so every existing QLS algorithm still requires exponential time on the rescaled system. Our algorithm does not, because its runtime depends on $\ET$, not on $\kappa(AD)$.

The cost of right-rescaling is that the quantum state now encodes $\zvec$, not the original solution $\xvec$. For polynomial systems this is acceptable: the downstream goal is to recover a Boolean assignment, and we show that this can be extracted from $\ket{\zvec}$ via $O(\log n)$ computational basis measurements (\cref{lem:solutionextractionlogn}). 

Combining the instance-aware QLS algorithm with right-rescaling, we obtain a new quantum algorithm for solving multivariate polynomial systems (\cref{alg:MQC}). For the rescaled system $AD\zvec = \bvec$, the algorithm achieves runtime $\widetilde{O}(\sqrt{\ET})$ (\cref{thm:qlspolytime}), where $\ET$ is defined with respect to the weighted Boolean Macaulay matrix. Table~\ref{tab:polysys} summarizes how the combined framework, instance-aware QLS plus right rescaling, compares to existing QLS-based approaches for solving polynomial systems.

\begin{table}[ht]
\centering
\renewcommand{\arraystretch}{1.4}
\small
\begin{tabular}{|c|c|c|}
\hline
 & \textbf{Existing QLS-based~\cite{ding2023limitations}}
 & \textbf{Our algorithm (\cref{alg:MQC})} \\
\hline
Input
  & $\ket{\bvec} = \ket{0}^{n\log m}$
  & $\ket{\Phi_b} = \ket{0}^{O(n\log mn)}$ \\
\hline
Sparsity
  & $A$ is $\mathrm{poly}(n)$-sparse
  & $H=[AD,-\bvec]$ is $\mathrm{poly}(n)$-sparse \\
\hline
Complexity
  & $\widetilde{O}(\kappa(A))$
  & $\widetilde{O}(\sqrt{\ET})$ \\
\hline
MIS Example 
  & $\kappa(A) = \Omega(2^{h/2})$
  & $\ET = \mathrm{poly}(n)$ under
    conditions of Thm.~\ref{thm:mainMIS} \\
\hline
Output
  & Recover $\bs\in\{0,1\}^n$ in $\mathrm{poly}(n)$ time
  & Recover $\bs\in\{0,1\}^n$ in $\mathrm{poly}(n)$ time \\
\hline
\end{tabular}
\caption{Comparison for solving polynomial systems via the
Boolean Macaulay linear system. Existing QLS algorithms
require exponential time due to
$\kappa(A) = \Omega(2^{h/2})$. Our algorithm, applied to
the rescaled system $AD\zvec = \bvec$, runs in polynomial
time whenever $\ET = \mathrm{poly}(n)$.}
\label{tab:polysys}
\end{table}

\paragraph{Application: planted maximum independent set.}

As a concrete application, we formulate the planted maximum independent set (MIS) problem as a multivariate polynomial system. Given a graph $G = (V, E)$ with $|V| = n$ containing a unique maximum independent set $S$ of size $h$, the goal is to recover $S$. We provide a detailed analysis of $\ET$ in terms of the independent-set structure of the graph. For every $1 \leq i \leq h$, let $I_i$ denote the number of independent sets of size $i$. If
\[
  I_i = \mathrm{poly}(n) \cdot \binom{h}{i}
  \quad \text{for all } 1 \leq i \leq h,
\]
then $\ET = \mathrm{poly}(n)$ and our algorithm recovers $S$ in polynomial time. No efficient classical algorithm is currently known under the same conditions. We do not claim that these conditions capture classically hard instances, establishing classical hardness here is an important open problem, but the algorithm's runtime is \emph{instance-dependent}, adapting to the independent-set structure of the specific input graph, whereas the best classical algorithms for MIS on random graphs are instance-insensitive and face barriers from the overlap gap property~\cite{gamarnik2021overlap}.

\paragraph{Precision.}
Our algorithm has polynomial dependence on $1/\epsilon$ (via phase estimation) rather than the logarithmic dependence achieved by QSVT-based solvers. For the polynomial system application, this is not a limitation: the downstream task is a computational basis measurement, and trace-distance error $\epsilon$ shifts each outcome probability by at most $\epsilon$~\cite[Theorem~9.1]{nielsen2002QCQI}. Setting $\epsilon = 1/\mathrm{poly}(n)$ preserves the recovery guarantee, so the total cost remains polynomial whenever $\ET = \mathrm{poly}(n)$. Whether the polynomial $\epsilon$-dependence can be improved to logarithmic, or if it is inherent in this approach, is an interesting open question.

\subsection{Algorithm Overview}\label{sec:overview}

We interpret the augmented matrix $H = [A, -\bvec]$ as the biadjacency matrix of a weighted bipartite graph: rows $\{u_1, \ldots, u_M\}$ form one vertex set, columns $\{v_1, \ldots, v_N, b\}$ form the other, and edges correspond to nonzero entries. From this graph we define \emph{row star states} $\ket{\Psi_i}$ (\cref{def:rowstates}) and \emph{column star states} $\ket{\Phi_j}$, $\ket{\Phi_b}$ (\cref{def:colstates}), and construct the walk operator
\[
  U_{\mathcal{AB}} = (2\Pi_{\mathcal{A}} - I)(2\Pi_{\mathcal{B}} - I),
\]
where $\Pi_{\mathcal{A}}$ projects onto the orthogonal complement of the column star states and $\Pi_{\mathcal{B}}$ projects onto the span of the row star states.

For graph-related linear systems~\cite{li2025multidimensional}, each column of $H$ has exactly two nonzero entries and the walk operator reduces to a $\pm\mathrm{SWAP}$ gate. For general sparse matrices, columns may have arbitrarily many nonzero entries, and this construction breaks down. We address this with two design choices: (i) each column star state distributes amplitude \emph{uniformly} over all nonzero entries of its column, regardless of their values; and (ii) each row star state is reweighted to reflect column sparsity. These choices preserve the spectral properties needed for the effective spectral gap lemma (\cref{lem:effective-spectral-gap}) while accommodating general sparsity patterns. The unitary $U_{\mathcal{AB}}$ can be implemented in time $\widetilde{O}(s)$ (\cref{thm:unitary-from-reflections}).

The algorithm proceeds in three steps:
\begin{enumerate}
  \item \textbf{Prepare} the initial state $\ket{\Phi_b} = \frac{1}{\sqrt{\mathrm{nnz}(\bvec)}} \sum_{b_i \neq 0} \ket{e_{i,b}}$. We require $\bvec$ to be $O(s)$-sparse; even under this constraint, the QLS problem remains BQP-complete~\cite[Section~4.3]{Gharibian2024QCT}.

  \item \textbf{Phase estimation} on $(U_{\mathcal{AB}}, \ket{\Phi_b})$ with precision $\delta = O(\epsilon^2 / \sqrt{s \cdot \ET})$. By \cref{thm:statepreparation}, measuring phase $0$ yields a state $\epsilon$-close in trace distance to
  \[
    \ket{\theta^*} = \frac{1}{\sqrt{1 + \|\xvec\|_2^2}}
    \Bigl(\sum_{j=1}^{N} \y_j \ket{\Phi_j} + \ket{\Phi_b}\Bigr),
  \]
  where $\xvec = [\y_1, \ldots, \y_N]^{\!\top} = A^{+}\bvec$. This relies on three conditions, initial state decomposition (\cref{lemm:statedecompositonwalk}), fixed-point property (\cref{lem:fixedpoint}), and projection property (\cref{lem:projectionB}), together with the effective spectral gap lemma (\cref{lem:effective-spectral-gap}).

  \item \textbf{Projective measurement} $\{I - \ket{\Phi_b}\!\bra{\Phi_b},\; \ket{\Phi_b}\!\bra{\Phi_b}\}$ on $\ket{\theta^*}$ yields $\ket{\xvec} = \frac{1}{\|\xvec\|_2} \sum_{j=1}^{N} \y_j \ket{\Phi_j}$.
\end{enumerate}

\subsection{Comparison to Prior Work}\label{sec:comparison}

Our algorithm shares a high-level structure with the span-program algorithm of~\cite{jeffery2023quantum} and the kernel-projection algorithm of~\cite{dalzell2024shortcut}; all three produce the minimum $\ell_2$-norm solution. The key differences are in the unitary operator and the resulting complexity.

The span-program approach~\cite{jeffery2023quantum} produces the correct quantum state but does not yield a time-efficient algorithm since implementing the span-program unitary is costly~\cite{cornelissen2020SpanProgramTime}. Our vector decomposition theorem (\cref{thm:vectordecomposition}) parallels the formulation in~\cite[Equations~(12)--(13)]{jeffery2023quantum}, so our algorithm can be viewed as an efficiently implementable generalization of the span program approach.

The kernel projection/reflection approach~\cite{dalzell2024shortcut} uses QSVT to achieve optimal complexity $\widetilde{O}(\kappa(A) \log(1/\epsilon))$. In \cref{append:QLSQSVT} we show that their technique also applies to our augmented matrix $H$, yielding an alternative algorithm with the same optimal complexity. However, applying QSVT directly to $H$ does \emph{not} produce an instance-aware algorithm: the effective spectral gap lemma, the key tool for bounding complexity by $\sqrt{\ET}$ instead of $\kappa(A)$, requires the walk operator $U_{\mathcal{AB}}$ and cannot be applied to a generic QSVT polynomial with respect to $H$. This underscores the necessity of our walk based construction.

\subsection{Main Results}

\begin{theorem}[Instance aware QLS algorithm]\label{thm:introQLSalgorithmcomplexity}
Let $A \in \mathbb{R}^{M \times N}$, let
$\bvec \in \mathrm{range}(A)\setminus\{0\}$ with $H = [A,\,-\bvec]$ being
$s$-sparse, let $\xvec = A^{+}\bvec$ and $\|\xvec\|\geq 1$, and let
\[
    \ket{\xvec}
    \;=\;
    \frac{1}{\|\xvec\|_2}\sum_{j=1}^{N} y_j\,\ket{\Phi_j}
\]
denote the encoding of the normalized solution in the orthonormal column
star states of \cref{sec:graph}.

Suppose we are given sparse access oracles for $A$ and $\bvec$
(\cref{defn:sparse_access_oracle}), a target precision
$\epsilon \in (0,1/2)$. With
probability at least $2/3$, \cref{alg:newQLS2} outputs a quantum state
$\ket{\xvec'}$ satisfying
\[
    \bigl\|\ket{\xvec'}\!\bra{\xvec'} - \ket{\xvec}\!\bra{\xvec}\bigr\|_1
    \;\leq\; \epsilon,
\]
with query complexity
$O\!\bigl(\|\xvec\|_2^2\,\sqrt{s \cdot \ET}\,/\,\epsilon^2\bigr)$
and time complexity
$\widetilde{O}\!\bigl(\|\xvec\|_2^2\,\sqrt{s^3 \cdot \ET}\,/\,\epsilon^2\bigr)$.
\end{theorem}

\begin{remark}[$\mathit{ET}$ vs.\ $\kappa(A)$]\label{rem:et-vs-kappa}
Under $\|A\|\le 1$ and $\|\bvec\|=1$, 
every linear system satisfies 
$\mathit{ET}\le 2\cdot\kappa(A)^4$. 
This follows from $d_i \le 2$ for each row of $H$ 
and $\|\pvec\|_2 \le \kappa(A)^2$.
\end{remark}

\begin{theorem}[Quantum algorithm for polynomial systems]\label{thm:mainpoly}
Let $\mathcal{F} \subseteq \mathbb{C}[x_1, \ldots, x_n]$ be a system of quadratic polynomials with a unique Boolean solution $\bs \in \{0,1\}^n$ of Hamming weight $h$. With probability at least $1 - \delta$, \cref{alg:MQC} recovers $\bs$ in time
\[
  \widetilde{O}\!\bigl(\sqrt{\ET} \cdot \log(1/\delta)\bigr),
\]
where $\ET$ is defined with respect to the weighted Boolean Macaulay linear system $AD\zvec = \bvec$. Under the conditions of Theorem~\ref{thm:mainMIS}, this is $\mathrm{poly}(n)$.
\end{theorem}

\begin{theorem}[Planted MIS in polynomial time]\label{thm:mainMIS}
Let $G = (V, E)$ with $|V| = n$ contain a unique maximum independent set of size $h$. If $I_i = \mathrm{poly}(n) \cdot \binom{h}{i}$ for all $1 \leq i \leq h$, then $\ET = \mathrm{poly}(n)$ and \cref{alg:MQC} recovers the planted set in polynomial time with high probability.
\end{theorem}

\subsection{Technical Contributions}

Our contributions have three components.

\begin{enumerate}

\item \emph{A new walk operator for general linear systems.}
The bottleneck for extending the multidimensional electrical
network framework~\cite{li2025multidimensional} beyond
graph-related systems is that graph-related systems have
exactly two nonzero entries per column (representing an
edge), captured naturally by a $\pm$SWAP operator. General
linear systems have arbitrarily many nonzero entries per
column, breaking this construction. We resolve this by
redefining $\Pi_{\mathcal{A}}$ to project onto the
orthogonal complement of \emph{uniformly weighted} column
star states, and reweighting the row star states by column
sparsity. This preserves the spectral properties needed for
the effective spectral gap
lemma~(\cref{lem:effective-spectral-gap}) while accommodating
arbitrary sparsity patterns, yielding a runtime that depends
on $\ET$ rather than $\kappa(A)$.

\item \emph{Right-rescaling for polynomial systems.}
Existing quantum preconditioning methods require the output
state to correspond to the solution of the original
system~\cite{clader2013preconditioned,shao2018quantum},
preserving the truncated QLS condition number lower bound. We
introduce a right-rescaling $AD\zvec = \bvec$ whose goal is
not to recover $\xvec$ but to solve a structurally modified
system amenable to our instance-aware algorithm. The
rescaling reduces the solution norm, bypassing the condition number barrier
of~\cite{ding2023limitations} for Boolean Macaulay systems
simultaneously.

\item \emph{End-to-end application to planted MIS.}
Our QLS-based approach provides a concrete runtime analysis
that depends on the independent-set structure of the input
graph (\cref{subsec:MISpolysys}): the algorithm runs in
polynomial time whenever the number of independent sets
of size~$i$ satisfies $I_i \le \mathrm{poly}(n)\binom{h}{i}$
for all~$i$, where $h$ is the planted set size.
A distinguishing feature is that the runtime is
\emph{instance-dependent}, it adapts to the specific graph, whereas
classical algorithms for MIS on random graphs are
instance-insensitive~\cite{gamarnik2021overlap}.
We do not currently know of a classically hard graph family
satisfying this condition; identifying one is an interesting
open problem.

\end{enumerate}

\subsection{Directions for Future Work}\label{sec:outlook}

Our instance-aware QLS algorithm, combined with the
right-rescaling scheme, reduces a broad class of problems to
bounding $\ET$ for a structured linear system. We outline
four directions, stating what remains open in each case.

\begin{enumerate}

\item \textbf{Graph isomorphism (GI).}
GI can be formulated as a polynomial system: two graphs on
$n$ vertices are isomorphic if and only if the corresponding
$M$-graph~\cite{kozen1978clique} contains a clique of size
$n$, which reduces to MIS in the complement graph. Alternatively,
GI can be encoded as a SAT
instance~\cite{toran2013resolution} and then as a polynomial
system. Our framework applies to either formulation.
\emph{What remains open:} designing a rescaling matrix $D$ (not necessarily diagonal) for the
GI-derived Macaulay system and bounding
$\ET$.

\item \textbf{Combinatorial optimization.}
The planted MIS analysis extends to QUBO problems
$\sum_{i<j} c_{ij} x_i x_j = t$ with $x_i \in \{0,1\}$
and a unique solution. More broadly, for random instances of
$k$-SAT and MIS, existing classical algorithms face barriers
from the overlap gap
property~\cite{gamarnik2021overlap}, and low-depth QAOA
exhibits similar
limitations~\cite{basso2022performance}. Our framework's
instance-dependent runtime offers a structural contrast to
these instance-insensitive approaches. \emph{What remains
open:} identifying hard instances with provably polynomial
$\ET$, and establishing that classical algorithms remain
inefficient under the same conditions.

\item \textbf{Post-quantum cryptography.}
Both multivariate-based and lattice-based cryptosystems
reduce to polynomial
systems~\cite{ding2005rainbow,chen2025quantum}. The
underlying lattice problems lie in
$\mathrm{NP} \cap \mathrm{coNP}$~\cite{aharonov2005lattice},
suggesting a quantum advantage may exist, yet little quantum
cryptanalysis has been
conducted~\cite{eldar2022efficient,chen2022quantum,chen2022quantumlwe,chen2024quantum}.
\emph{What remains open:} bounding $\ET$ for the Macaulay
systems derived from concrete cryptographic problems. A
polynomial bound would yield a new quantum cryptanalytic
approach for the corresponding schemes.

\item \textbf{Nonlinear differential equations.}
Carleman linearization closely parallels Macaulay
linearization~\cite{liu2021efficient,krovi2023improved,costa2025further}:
both require solving exponentially large linear systems via
QLS. Recently, \cite{costa2025further} proposed a rescaling
scheme for the Carleman setting but noted it inflates the
condition number exponentially, and suggested an alternative
QLS algorithm may be needed. Our instance-aware algorithm is
a natural candidate. \emph{What remains open:} analyzing
$\ET$ for Carleman linear systems arising from specific nonlinear differential equations
such as Navier Stokes.

\end{enumerate}

In each case, the challenge is designing $D$ and proving
$\ET = \mathrm{poly}(n)$; our MIS analysis
(\cref{subsec:MISpolysys}) serves as a template.

A separate direction is adapting our instance-aware
techniques to ground-state preparation for structured
Hamiltonians. Current algorithms prepare the ground state
of a Hamiltonian $H$ in time scaling with the inverse
spectral gap $\Delta^{-1}$~\cite{ge2019faster,lin2020near},
mirroring how prior QLS algorithms scale with $\kappa(A)$.
The structural parallel is direct: just as our QLS algorithm
replaces $\kappa(A)$ with an instance-aware parameter
$\mathit{ET}$ that captures how $\bvec$ aligns with $A$'s
eigenspaces, one could seek a parameter for ground-state
preparation that captures how the guiding state
$\ket{\psi_0}$ overlaps with the eigenspaces of $H$.
\emph{What remains open:} identifying Hamiltonian families where
the resulting instance-aware parameter is provably smaller
than $\Delta^{-1}$.

\subsection{Organization}

\cref{sec:linearalgebra} introduces the linear algebra background and the vector decomposition theorem. \cref{sec:QLSalgorithm} presents the new QLS algorithm, including the complexity analysis.  \cref{sec:examples} demonstrates that $\mathit{ET}$ can be
exponentially smaller than $\kappa(A)$ via two separation
witnesses: a diagonal linear system and a welded tree linear system instance. \cref{sec:newalgpolysys} develops the rescaling scheme and the quantum algorithm for polynomial systems, including the application to the planted MIS.

\section{Linear Algebra Background}\label{sec:linearalgebra}

We introduce the linear system problem and its quantum
counterpart, then prove the vector decomposition theorem
that underpins our QLS algorithm.

\begin{problem}[QLSP]\label{prob:QLSP}
Given $A \in \mathbb{R}^{M \times N}$ and
$\bvec \in \mathbb{R}^{M}$, output a quantum state
$\ket{\xvec'}$ such that
$\|\ket{\xvec'}\!\bra{\xvec'} -
\ket{\xvec}\!\bra{\xvec}\|_1 \leq \epsilon$,
where
$\ket{\xvec} = \frac{1}{\|\xvec\|_2}
\sum_{i=1}^{N} \y_i \ket{i}$ and $\xvec = A^{+}\bvec$.
\end{problem}

Let $H = [A,\, -\bvec] \in \mathbb{R}^{M \times (N+1)}$ be
the augmented matrix of the linear system $A\xvec = \bvec$,
with rows $\rvec_1, \ldots, \rvec_M$.

\begin{definition}[Row space and null space of $H$]
\label{def:spaces}
The row space of $H$ is
$\mathrm{Row}(H) =
\mathrm{span}\{\rvec_1, \ldots, \rvec_M\}
\subseteq \mathbb{R}^{N+1}$.
The null space is
$\mathrm{Null}(H) =
\{\vec{u} \in \mathbb{R}^{N+1} \mid H\vec{u} = 0\}$.
\end{definition}

Every vector $\vec{u} \in \mathrm{Null}(H)$ satisfies
$\langle \rvec_i, \vec{u} \rangle = 0$ for all $i$, and
hence $\langle \vec{v}, \vec{u} \rangle = 0$ for every
$\vec{v} \in \mathrm{Row}(H)$.

\begin{lemma}[Fundamental Theorem of Linear Algebra
{\cite{strang2022introduction}}]\label{lem:ranknull}
The row space and the null space are orthogonal complements
in $\mathbb{R}^{N+1}$:
\[
  \mathrm{Row}(H) \oplus \mathrm{Null}(H)
  = \mathbb{R}^{N+1}.
\]
Consequently, every
$\vec{\psi}_0 \in \mathbb{R}^{N+1}$ decomposes uniquely as
$\vec{\psi}_0 = \vec{u} + \vec{v}$ with
$\vec{u} \in \mathrm{Null}(H)$,
$\vec{v} \in \mathrm{Row}(H)$, and
$\langle \vec{u}, \vec{v} \rangle = 0$.
\end{lemma}

\begin{theorem}[Vector Decomposition]
\label{thm:vectordecomposition}
Let $A \in \mathbb{R}^{M \times N}$, $\bvec \in \mathrm{range}(A)$, and let
$H = [A,\, -\bvec] \in \mathbb{R}^{M \times (N+1)}$. Set
$\xvec = A^{+}\bvec$, the minimum $\ell_2$-norm solution of $A\xvec = \bvec$,
and $\vec{\psi}_0 = [0, \ldots, 0, 1]^{\!\top} \in \mathbb{R}^{N+1}$. Then
\begin{equation}\label{eq:decomp}
  \vec{\psi}_0
  = \frac{1}{1 + \|\xvec\|_2^2}
  \left(
    \underbrace{\begin{bmatrix}\xvec\\1\end{bmatrix}}_{\vec{\theta}\,\in\,\mathrm{Null}(H)}
    -\underbrace{\begin{bmatrix}\xvec\\ -\|\xvec\|_2^2\end{bmatrix}}_{\vec{\theta}^{\perp}\,\in\,\mathrm{Row}(H)}
  \right),
  \qquad \vec{\theta} \perp \vec{\theta}^{\perp},
\end{equation}
where explicitly $\vec{\theta}^{\perp} = H^{\!\top}\pvec$ for
$\pvec = (AA^{\!\top})^{+}\bvec$, which satisfies $A^{\!\top}\pvec = \xvec$.

Conversely, if $\xvec$ is any solution of $A\xvec = \bvec$ with
$[\xvec;\, -\|\xvec\|_2^2] \in \mathrm{Row}(H)$, then $\xvec = A^{+}\bvec$.

\end{theorem}

\begin{proof}
We verify that the two components belong to the claimed
subspaces. Define
$\vec{\theta} = [\xvec^{\!\top},\, 1]^{\!\top}$. Then
\[
  H\vec{\theta}
  = [A,\, -\bvec]
    \begin{bmatrix}\xvec\\1\end{bmatrix}
  = A\xvec - \bvec = 0,
\]
so $\vec{\theta} \in \mathrm{Null}(H)$.

Next, define
$\vec{\theta}^{\perp} =
[\xvec^{\!\top},\, -\|\xvec\|_2^2]^{\!\top}$.
To show $\vec{\theta}^{\perp} \in \mathrm{Row}(H)$, we
exhibit a vector $\pvec$ such that
$H^{\!\top}\pvec = \vec{\theta}^{\perp}$. Let
\begin{equation}\label{eq:potential}
  \pvec = (AA^{\!\top})^{+}\bvec.
\end{equation}
Computing each block separately:
\begin{align*}
  A^{\!\top}\pvec
    &= A^{\!\top}(AA^{\!\top})^{+}\bvec
    = A^{+}\bvec = \xvec, \\
  -\bvec^{\!\top}\pvec
    &= -\bvec^{\!\top}(AA^{\!\top})^{+}\bvec
    = -\|\xvec\|_2^2.
\end{align*}
Thus
$H^{\!\top}\pvec =
[\xvec^{\!\top},\, -\|\xvec\|_2^2]^{\!\top}
= \vec{\theta}^{\perp}$,
confirming $\vec{\theta}^{\perp} \in \mathrm{Row}(H)$.

The subtraction gives
\[
  \vec{\theta} - \vec{\theta}^{\perp}
  = \begin{bmatrix}\xvec\\1\end{bmatrix}
  - \begin{bmatrix}\xvec\\-\|\xvec\|_2^2\end{bmatrix}
  = \begin{bmatrix}0\\1 + \|\xvec\|_2^2\end{bmatrix}.
\]
Dividing by $1 + \|\xvec\|_2^2$
yields~\eqref{eq:decomp}. The decomposition is unique
by \cref{lem:ranknull}.

For the converse, let $\xvec'$ solve $A\xvec' = \bvec$ with
$[\xvec'^{\!\top},\, -\|\xvec'\|_2^2]^{\!\top} = H^{\!\top}\vec{q}
= [\,(A^{\!\top}\vec{q})^{\!\top},\, -\bvec^{\!\top}\vec{q}\,]^{\!\top}$
for some $\vec{q} \in \mathbb{R}^{M}$. The first block gives
$\xvec' = A^{\!\top}\vec{q} \in \mathrm{Row}(A)$. Since
$\xvec = A^{+}\bvec \in \mathrm{Row}(A)$ as well and
$A(\xvec' - \xvec) = \bvec - \bvec = 0$, we have
$\xvec' - \xvec \in \mathrm{Row}(A) \cap \mathrm{Null}(A) = \{0\}$,
so $\xvec' = A^{+}\bvec$.
\end{proof}

\begin{remark}\label{remark:pby}
By definition of $\pvec = (AA^{\!\top})^{+}\bvec$,
\begin{equation}\label{eq:psolutionnorm}
  \pvec^{\!\top}\bvec
  = \bvec^{\!\top}(AA^{\!\top})^{+}\bvec
  = \|A^{+}\bvec\|_2^2
  = \|\xvec\|_2^2.
\end{equation}
By the Cauchy--Schwarz inequality,
$\|\pvec\| \cdot \|\bvec\| \geq \|\xvec\|_2^2$.
\end{remark}

\section{The New QLS Algorithm}\label{sec:QLSalgorithm}

We now present our quantum algorithm for solving \cref{prob:QLSP}.
The algorithm proceeds in three stages:
(i)~interpret the augmented matrix $H = [A,\,-\bvec]$ as the biadjacency matrix of a weighted bipartite graph~(\cref{sec:graph});
(ii)~use the graph structure to define a quantum walk operator $U_{\mathcal{AB}}$ and an initial state $\ket{\Phi_b}$~(\cref{sec:algorithmsetup});
(iii)~apply quantum phase estimation to the pair $(U_{\mathcal{AB}},\,\ket{\Phi_b})$ and post-select to extract the solution state~(\cref{sec:algorithm}).

The key departure from prior QLS algorithms is the choice of unitary.
Whereas QSVT-based solvers \cite{gilyen2018QSingValTransf} apply a polynomial transformation to the singular values of $A$, incurring a cost proportional to $\kappa(A)$, our operator $U_{\mathcal{AB}}$ is a quantum walk on the bipartite graph $G$ whose edge weights encode the entries of $[A,\,-\bvec]$ \emph{rescaled by the local degree structure}.
Phase estimation on $U_{\mathcal{AB}}$ therefore depends on the \emph{effective spectral gap} and the resulting complexity is governed by a new instance-dependent parameter
\[
    \ET \;=\; \sum_{i=1}^{M} p_i^{\,2}\, d_i,
    \qquad
    \text{where } \pvec = (AA^{\!\top})^{+}\bvec
    \;\;\text{and}\;\;
    d_i = \|H_{i,\cdot}\|_2^2,
\]
in place of the condition number $\kappa(A)$.

\begin{remark}[Intuition for $\ET$]\label{rem:ET-intuition}
The parameter $\ET$ originates from the theory of (multidimensional) quantum walks on (multidimensional) electrical networks~\cite{apers2022elfs,jeffery2023multidimensional,li2025multidimensional}, where it corresponds to the \emph{escape time} (hence the abbreviation).
Concretely, $\pvec$ records how the right-hand side $\bvec$ distributes across the row space of $A$: if $\bvec$ is well-aligned with the high-energy eigenspaces of $AA^{\!\top}$, then $\pvec$ is short and $\ET$ is small.
The weights $d_i$ account for the row norms of the augmented matrix.
Unlike $\kappa(A)$, which is a property of $A$ alone, $\ET$ depends jointly on $(A,\bvec)$ and can be much smaller, up to exponentially so, when the instance has favorable structure (see \cref{sec:examples}).
At the same time, $\ET \le \|\pvec\|_2^2 \cdot \max_i d_i \le \kappa(A)^4\cdot \max_i d_i$, so $\kappa(A)$ remains a valid worst-case upper bound when $\max_i d_i$ bounded by a constant, consistent with standard perturbation theory.
\end{remark}

\subsection{Graph Construction}\label{sec:graph}

Given an $M \times (N\!+\!1)$ matrix $H = [A,\,-\bvec]$ that is $s$-sparse (every row and every column has at most $s$ nonzero entries), we construct a weighted bipartite graph $G = (V, E, w)$ as follows.

\begin{enumerate}[label=(\roman*)]
    \item \textbf{Vertices.}\;
    Introduce one vertex $u_i$ for each row $i \in [M]$, one vertex $v_j$ for each column $j \in [N]$ of $A$, and a distinguished vertex $b$ representing the column $-\bvec$.
    The vertex set is $V = \{u_1,\dots,u_M\} \cup \{v_1,\dots,v_N\} \cup \{b\}$.

    \item \textbf{Edges and weights.}\;
    For every nonzero entry $A_{i,j} \neq 0$, add an edge $e_{i,j} = \{u_i, v_j\}$ with weight
    \[
        w_{i,j} \;=\; |A_{i,j}|^2 \,|\mathcal{N}(v_j)|,
    \]
    where $\mathcal{N}(v_j)$ denotes the neighborhood of $v_j$ in $G$.
    For every nonzero entry $b_i \neq 0$, add an edge $e_{i,b} = \{u_i, b\}$ with weight
    \[
        w_{i,b} \;=\; |b_i|^2\,|\mathcal{N}(b)|.
    \]
\end{enumerate}

\noindent
The edge set is $E = \{e_{i,j} : A_{i,j} \neq 0\} \cup \{e_{i,b} : b_i \neq 0\}$.
Since $H$ is $s$-sparse, each vertex has degree at most $s$ in $G$.

\begin{example}\label{ex:graph}
Let $H = [A , -\bvec] = \begin{bsmallmatrix} 1 & 1 & 0 & 1 \\ 0 & 1 & 0 & 0 \\ 0 & 0 & 1 & 1 \end{bsmallmatrix}$.
The graph $G$ has row vertices $\{u_1, u_2, u_3\}$, column vertices $\{v_1, v_2, v_3\}$, and the distinguished vertex $b$.
The edges are $\{u_1,v_1\}, \{u_1,v_2\}, \{u_2,v_2\}, \{u_3,v_3\}$ from $A$, and $\{u_1,b\}, \{u_3,b\}$ from $\bvec$.
The edge weights are $w_{1,1} = 1$, $w_{1,2} = w_{2,2} = 2$, $w_{3,3} = 1$, $w_{1,b} = w_{3,b} = 2$.
See \cref{fig:graph-construction}.
\end{example}

\begin{figure}[ht]
    \centering
    \begin{tikzpicture}[
        scale=1.3,
        unode/.style={draw, circle, fill=blue!15, minimum size=8mm, inner sep=0pt, font=\small},
        vnode/.style={draw, circle, fill=green!15, minimum size=8mm, inner sep=0pt, font=\small},
        bnode/.style={draw, circle, fill=red!15, minimum size=8mm, inner sep=0pt, font=\small},
        every edge/.style={thick}
    ]
        \foreach \i in {1, 2, 3} {
            \node[unode] (u\i) at (0, -1.3*\i) {$u_{\i}$};
        }
        \foreach \j in {1, 2, 3} {
            \node[vnode] (v\j) at (3.2, -1.3*\j) {$v_{\j}$};
        }
        \node[bnode] (b) at (6.4, -2.6) {$b$};

        \draw (u1) -- node[above, font=\scriptsize] {$1$} (v1);
        \draw (u1) -- node[above, font=\scriptsize, pos=0.35] {$2$} (v2);
        \draw (u2) -- node[above, font=\scriptsize] {$2$} (v2);
        \draw (u3) -- node[above, font=\scriptsize] {$1$} (v3);

        \draw[dashed] (u1) -- node[above, font=\scriptsize, pos=0.6] {$2$} (b);
        \draw[dashed] (u3) -- node[below, font=\scriptsize, pos=0.6] {$2$} (b);
    \end{tikzpicture}
    \caption{The graph $G$ for \cref{ex:graph}. Blue nodes are row vertices ($u_i$), green nodes are column vertices ($v_j$), the red node is $b$. Solid edges correspond to nonzero entries of $A$; dashed edges to nonzero entries of $\bvec$. Labels indicate edge weights.}
    \label{fig:graph-construction}
\end{figure}

\subsection{Algorithm Setup}\label{sec:algorithmsetup}

We now define the Hilbert space, the subspaces, and the three key quantum states that underpin the algorithm.

\paragraph{Edge space and star states.}
The Hilbert space is
\[
    \mathcal{H} \;=\; \operatorname{span}\bigl\{\ket{e_{i,j}} : \{u_i, v_j\} \in E\bigr\}
    \;\oplus\;
    \operatorname{span}\bigl\{\ket{e_{i,b}} : \{u_i, b\} \in E\bigr\},
\]
where each edge of $G$ is associated with a unique orthonormal basis vector.
Edges carry a sign determined by the corresponding matrix entry:
for $\{u_i, v_j\} \in E$, set $\Delta_{u_i, v_j} = 0$ if $A_{i,j} > 0$ and $\Delta_{u_i, v_j} = 1$ if $A_{i,j} < 0$; similarly, $\Delta_{u_i, b} = 0$ if $-b_i > 0$ and $\Delta_{u_i, b} = 1$ otherwise.

For each vertex $u_i$, the \emph{weighted degree} is
\[
    d_{u_i} \;=\; \sum_{v_j \in \mathcal{N}(u_i)} w_{i,j} \;+\; w_{i,b}.
\]

\begin{definition}[Row star states]\label{def:rowstates}
For each row vertex $u_i$, the \emph{row star state} is
\begin{equation}\label{eq:ithrowH}
    \ket{\Psi_i} \;=\; \frac{1}{\sqrt{d_{u_i}}}
    \biggl(
        \sum_{v_j \in \mathcal{N}(u_i)} (-1)^{\Delta_{u_i,v_j}} \sqrt{w_{i,j}}\;\ket{e_{i,j}}
        \;+\; (-1)^{\Delta_{u_i,b}} \sqrt{w_{i,b}}\;\ket{e_{i,b}}
    \biggr).
\end{equation}
\end{definition}

\begin{definition}[Column star states]\label{def:colstates}
For each column vertex $v_j$, the \emph{column star state} is
\[
    \ket{\Phi_j} \;=\; \frac{1}{\sqrt{|\mathcal{N}(v_j)|}}
    \sum_{u_i \in \mathcal{N}(v_j)} \ket{e_{i,j}}.
\]
For the vertex $b$, the column star state is
\[
    \ket{\Phi_b} \;=\; \frac{1}{\sqrt{|\mathcal{N}(b)|}}
    \sum_{u_i \in \mathcal{N}(b)} \ket{e_{i,b}}.
\]
\end{definition}

\noindent
The row star states $\{\ket{\Psi_i}\}_{i=1}^{M}$ are orthonormal since distinct rows share no edges; similarly, the column star states $\{\ket{\Phi_j}\}_{j=1}^{N} \cup \{\ket{\Phi_b}\}$ are orthonormal.

A key identity connects the star states to the matrix entries:
for any edge $\{u_i, v_j\} \in E$,
\begin{equation}\label{eq:inner-product-identity}
    \langle \Psi_i \mid \Phi_j \rangle
    \;=\;
    \frac{A_{i,j}}{\sqrt{d_{u_i}}},
    \qquad\qquad
    \langle \Psi_i \mid \Phi_b \rangle
    \;=\;
    \frac{-b_i}{\sqrt{d_{u_i}}}.
\end{equation}
This follows directly from the definitions: $(-1)^{\Delta_{u_i,v_j}} \sqrt{w_{i,j}/|\mathcal{N}(v_j)|} = A_{i,j}$ and $(-1)^{\Delta_{u_i,b}} \sqrt{w_{i,b}/|\mathcal{N}(b)|} = -b_i$.
These inner products will be used repeatedly in the proofs below.

\paragraph{Subspaces and the walk operator.}
Define two subspaces of $\mathcal{H}$:
\begin{align}
    \mathcal{A}^{\perp} &\;=\; \operatorname{span}\bigl(\{\ket{\Phi_j}\}_{j=1}^{N} \cup \{\ket{\Phi_b}\}\bigr),
    &
    \mathcal{A} &\;=\; (\mathcal{A}^{\perp})^{\perp}, \label{eq:subA} \\[4pt]
    \mathcal{B} &\;=\; \operatorname{span}\bigl(\{\ket{\Psi_i}\}_{i=1}^{M}\bigr),
    &&  \label{eq:subB}
\end{align}
with orthogonal projectors $\Pi_{\mathcal{A}}$ and $\Pi_{\mathcal{B}}$ onto $\mathcal{A}$ and $\mathcal{B}$, respectively.
Explicitly,
\begin{equation}\label{eq:projectors}
    \Pi_{\mathcal{A}} \;=\; I - \sum_{j=1}^{N} \ket{\Phi_j}\!\bra{\Phi_j} - \ket{\Phi_b}\!\bra{\Phi_b},
    \qquad
    \Pi_{\mathcal{B}} \;=\; \sum_{i=1}^{M} \ket{\Psi_i}\!\bra{\Psi_i}.
\end{equation}
The \emph{quantum walk operator} is the product of reflections
\begin{equation}\label{eq:unitaryAB}
    U_{\mathcal{AB}} \;=\; (2\Pi_{\mathcal{A}} - I)\,(2\Pi_{\mathcal{B}} - I).
\end{equation}

\paragraph{Three key quantum states.}
Recall that $\pvec = (AA^{\!\top})^{+}\bvec$ satisfies $A^{\!\top}\pvec = \xvec$ and $\pvec^{\!\top}\bvec = \|\xvec\|_2^2$.
We define:
\begin{enumerate}[label=(\alph*)]
    \item \textbf{Initial state:}\;
    $\ket{\psi_0^*} = \ket{\Phi_b}$.

    \item \textbf{Fixed-point state:}
    \begin{equation}\label{eq:solutionstate}
        \ket{\theta^*} \;=\; \frac{1}{\sqrt{1 + \|\xvec\|_2^2}}
        \biggl(\sum_{j=1}^{N} y_j\,\ket{\Phi_j} \;+\; \ket{\Phi_b}\biggr).
    \end{equation}

    \item \textbf{Potential state:}
    \begin{equation}\label{eq:potential}
        \ket{\mathfrak{p}} \;=\; \frac{1}{1 + \|\xvec\|_2^2}
        \sum_{i=1}^{M} p_i\,\sqrt{d_{u_i}}\;\ket{\Psi_i}.
    \end{equation}
\end{enumerate}

\noindent
The state $\ket{\theta^*}$ encodes the solution $\xvec$ in the column-star-state basis; extracting $\xvec$ from $\ket{\theta^*}$ requires only a projective measurement removing the $\ket{\Phi_b}$ component (see Step~3 of \cref{alg:newQLS2}).
The state $\ket{\mathfrak{p}}$ is unnormalized and lies entirely in $\mathcal{B}$; its squared norm controls the precision required for phase estimation.

\subsubsection{Verification of the Phase Estimation Framework}

Our algorithm rests on the following general result, which gives conditions under which phase estimation on $U_{\mathcal{AB}}$ extracts the fixed-point state from the initial state.
The statement below is from~\cite[Lemma~2.12]{li2025multidimensional}, which is a refinement of~\cite[Lemma~8]{piddock2019electricfind} and~\cite[Lemma~10]{apers2022elfs}; we refer the reader the full proof in ~\cite[Appendix~A]{li2025multidimensional} for completeness.

\begin{lemma}[Phase estimation extraction~{\cite[Lemma~2.12]{li2025multidimensional}}]\label{lem:phaseprojection}
Let $U_{\mathcal{AB}} = (2\Pi_{\mathcal{A}} - I)(2\Pi_{\mathcal{B}} - I)$ act on a Hilbert space $\mathcal{H}$, for projectors $\Pi_{\mathcal{A}},\Pi_{\mathcal{B}}$ onto subspaces $\mathcal{A},\mathcal{B} \subseteq \mathcal{H}$.
Suppose there exist a unit vector $\ket{\theta^*}$, a scalar $\gamma \in (0,1]$, and an unnormalized vector $\ket{\mathfrak{p}}$ such that:
\begin{enumerate}[label=\textup{(C\arabic*)}]
    \item\label{cond:decomp} \textit{Initial state decomposition:}\; $\ket{\psi_0^*} = \sqrt{\gamma}\,\ket{\theta^*} - (I - \Pi_{\mathcal{A}})\ket{\mathfrak{p}}$.
    \item\label{cond:fixed} \textit{Fixed-point condition:}\; $U_{\mathcal{AB}}\ket{\theta^*} = \ket{\theta^*}$.
    \item\label{cond:proj} \textit{Projection property:}\; $\Pi_{\mathcal{B}}\ket{\mathfrak{p}} = \ket{\mathfrak{p}}$.
\end{enumerate}
Then performing phase estimation on $\ket{\psi_0^*}$ with operator $U_{\mathcal{AB}}$ and precision $\delta$ and post-selecting on the eigenvalue $1$ (phase value $0$) succeeds with probability
$p \in \bigl[\tfrac{4}{\pi^2}\gamma,\;\gamma + \tfrac{17\pi^2}{16}\delta\|\ket{\mathfrak{p}}\|\bigr]$,
and the resulting state $\ket{\theta'}$ satisfies
\[
    \tfrac{1}{2}\bigl\|\ket{\theta'}\!\bra{\theta'} - \ket{\theta^*}\!\bra{\theta^*}\bigr\|_1
    \;\leq\;
    \sqrt{\frac{17\pi^2\,\delta\,\|\ket{\mathfrak{p}}\|}{16\,\gamma}}.
\]
In particular, choosing $\delta = O\!\bigl(\gamma\epsilon^2 / \|\ket{\mathfrak{p}}\|\bigr)$ yields $\frac{1}{2}\|\ket{\theta'}\!\bra{\theta'} - \ket{\theta^*}\!\bra{\theta^*}\|_1 \leq \epsilon$.
\end{lemma}

\noindent
The proof relies on the effective spectral gap lemma of Lee, Magniez, and Santha~\cite{lee2011QQueryCompStateConv}, which we state for reference.

\begin{lemma}[Effective spectral gap lemma~{\cite{lee2011QQueryCompStateConv}}]\label{lem:effective-spectral-gap}
Fix $\delta \in [0,\pi)$ and let $\Pi_{\delta}$ be the projector onto the eigenspaces of $U_{\mathcal{AB}}$ with eigenvalue $e^{i\theta}$ satisfying $|\theta| \leq \delta$.
If $\ket{\mathfrak{p}} \in \mathcal{B}$, then
$\|\Pi_{\delta}(I - \Pi_{\mathcal{A}})\ket{\mathfrak{p}}\| \leq \tfrac{\delta}{2}\|\ket{\mathfrak{p}}\|$.
\end{lemma}

\noindent
The effective spectral gap lemma follows from Jordan's lemma: any two projectors jointly decompose the Hilbert space into invariant subspaces of dimension at most two, and within each such subspace the claim reduces to a direct trigonometric calculation.

\medskip
We now verify that our concrete definitions satisfy conditions~\ref{cond:decomp}--\ref{cond:proj}.

\begin{lemma}[Initial state decomposition]\label{lemm:statedecompositonwalk}
The initial state $\ket{\Phi_b}$ decomposes as
\[
    \ket{\Phi_b}
    \;=\;
    \frac{1}{\sqrt{1 + \|\xvec\|_2^2}}\;\ket{\theta^*}
    \;-\;
    (I - \Pi_{\mathcal{A}})\,\ket{\mathfrak{p}}.
\]
\end{lemma}

\begin{proof}
We compute $(I - \Pi_{\mathcal{A}})\ket{\mathfrak{p}}$ by expanding $\ket{\mathfrak{p}}$ in the row star states and projecting onto $\mathcal{A}^{\perp}$.
Using the inner-product identity~\eqref{eq:inner-product-identity}:
\begin{align*}
    (I - \Pi_{\mathcal{A}})\ket{\mathfrak{p}}
    &= \sum_{j=1}^{N} \ket{\Phi_j}\!\bra{\Phi_j}\ket{\mathfrak{p}}
       \;+\; \ket{\Phi_b}\!\bra{\Phi_b}\ket{\mathfrak{p}} \\[4pt]
    &= \frac{1}{1 + \|\xvec\|_2^2}
       \sum_{i=1}^{M} p_i
       \biggl(
           \sum_{v_j \in \mathcal{N}(u_i)} A_{i,j}\,\ket{\Phi_j}
           \;-\; b_i\,\ket{\Phi_b}
       \biggr) \\[4pt]
    &= \frac{1}{1 + \|\xvec\|_2^2}
       \biggl(
           \sum_{j=1}^{N} (A^{\!\top}\pvec)_j\,\ket{\Phi_j}
           \;-\; \pvec^{\!\top}\bvec\;\ket{\Phi_b}
       \biggr) \\[4pt]
    &= \frac{1}{1 + \|\xvec\|_2^2}
       \biggl(
           \sum_{j=1}^{N} y_j\,\ket{\Phi_j}
           \;-\; \|\xvec\|_2^2\;\ket{\Phi_b}
       \biggr),
\end{align*}
where the last step uses $A^{\!\top}\pvec = \xvec$ and $\pvec^{\!\top}\bvec = \|\xvec\|_2^2$.
Rearranging:
\[
    \ket{\Phi_b} + (I - \Pi_{\mathcal{A}})\ket{\mathfrak{p}}
    \;=\;
    \frac{1}{1 + \|\xvec\|_2^2}
    \biggl(\sum_{j=1}^{N} y_j\,\ket{\Phi_j} + \ket{\Phi_b}\biggr)
    \;=\;
    \frac{1}{\sqrt{1 + \|\xvec\|_2^2}}\;\ket{\theta^*}.
    \qedhere
\]
\end{proof}

\begin{lemma}[Fixed-point condition]\label{lem:fixedpoint}
$U_{\mathcal{AB}}\,\ket{\theta^*} = \ket{\theta^*}$.
\end{lemma}

\begin{proof}
It suffices to show $\Pi_{\mathcal{A}}\ket{\theta^*} = 0$ and $\Pi_{\mathcal{B}}\ket{\theta^*} = 0$, since then both reflections act as $-I$ and $(2 \cdot 0 - I)(2 \cdot 0 - I)\ket{\theta^*} = \ket{\theta^*}$.

\smallskip\noindent
\emph{Claim 1: $\Pi_{\mathcal{A}}\ket{\theta^*} = 0$.}\;
By definition~\eqref{eq:solutionstate}, $\ket{\theta^*}$ is a linear combination of $\{\ket{\Phi_j}\}$ and $\ket{\Phi_b}$, hence $\ket{\theta^*} \in \mathcal{A}^{\perp}$.

\smallskip\noindent
\emph{Claim 2: $\Pi_{\mathcal{B}}\ket{\theta^*} = 0$.}\;
Using the inner-product identity~\eqref{eq:inner-product-identity}:
\begin{align*}
    \bra{\Psi_i}\ket{\theta^*}
    &= \frac{1}{\sqrt{1 + \|\xvec\|_2^2}}
       \biggl(\sum_{v_j \in \mathcal{N}(u_i)} y_j \cdot \frac{A_{i,j}}{\sqrt{d_{u_i}}}
       \;-\; \frac{b_i}{\sqrt{d_{u_i}}}\biggr) \\[2pt]
    &= \frac{(A\xvec)_i - b_i}{\sqrt{(1 + \|\xvec\|_2^2)\,\sqrt{d_{u_i}}}}
    \;=\; 0,
\end{align*}
where the last equality follows from $A\xvec = \bvec$.
Since this holds for every $i$, we have $\Pi_{\mathcal{B}}\ket{\theta^*} = 0$.
\end{proof}

\begin{lemma}[Projection property]\label{lem:projectionB}
$\Pi_{\mathcal{B}}\,\ket{\mathfrak{p}} = \ket{\mathfrak{p}}$.
\end{lemma}

\begin{proof}
By definition~\eqref{eq:potential}, $\ket{\mathfrak{p}}$ is a linear combination of $\{\ket{\Psi_i}\}_{i=1}^{M}$, which span $\mathcal{B}$.
\end{proof}

\subsubsection{State Preparation via Phase Estimation}

Combining the three lemmas with~\cref{lem:phaseprojection}, we obtain:

\begin{theorem}[State preparation]\label{thm:statepreparation}
Let $U_{\mathcal{AB}}$ be as in~\eqref{eq:unitaryAB} and $\ket{\Phi_b}$ the initial state.
Performing phase estimation on $(U_{\mathcal{AB}},\,\ket{\Phi_b})$ with precision
$\delta = O\!\biggl(\frac{\epsilon^2}{\sqrt{s \cdot \ET}}\biggr)$
and post-selecting on phase value~$0$ succeeds with probability at least $\frac{1}{1 + \|\xvec\|_2^2}$ and produces a state $\ket{\theta'}$ satisfying
\[
    \tfrac{1}{2}\bigl\|\ket{\theta'}\!\bra{\theta'} - \ket{\theta^*}\!\bra{\theta^*}\bigr\|_1
    \;\leq\; \epsilon.
\]
The number of queries to controlled-$U_{\mathcal{AB}}$ is $O\!\bigl(\sqrt{s \cdot \ET}\,/\,\epsilon^2\bigr)$.
\end{theorem}

\begin{proof}
The three conditions of \cref{lem:phaseprojection} are verified by
\cref{lemm:statedecompositonwalk} (with $\gamma = 1/(1 + \|\xvec\|_2^2)$),
\cref{lem:fixedpoint}, and
\cref{lem:projectionB}.
The required precision is
$\delta = O(\gamma\epsilon^2 / \|\ket{\mathfrak{p}}\|)$,
and the number of queries to controlled-$U_{\mathcal{AB}}$ is $O(1/\delta)$.

It remains to bound $\|\ket{\mathfrak{p}}\|$.
Since $d_{u_i} \leq s\,d_i$ (where $d_i = \sum_{j} |A_{i,j}|^2 + |b_i|^2 = \|H_{i,\cdot}\|_2^2$ is the squared row norm of $H$ and $s$ is the column sparsity bound), we have
\[
    \|\ket{\mathfrak{p}}\|^2
    \;=\;
    \frac{1}{(1 + \|\xvec\|_2^2)^2}\sum_{i=1}^{M} p_i^2\,d_{u_i}
    \;\leq\;
    \frac{s \cdot \ET}{(1 + \|\xvec\|_2^2)^2}.
\]
Substituting into the query count:
\[
    O\!\biggl(\frac{1}{\delta}\biggr)
    \;=\;
    O\!\biggl(\frac{(1 + \|\xvec\|_2^2)\,\|\ket{\mathfrak{p}}\|}{\epsilon^2}\biggr)
    \;\leq\;
    O\!\biggl(\frac{\sqrt{s \cdot \ET}}{\epsilon^2}\biggr).
    \qedhere
\]
\end{proof}

\subsubsection{Efficient Implementation of $U_{\mathcal{AB}}$}\label{sec:implementation}

To translate query complexity into time complexity, we show that each application of $U_{\mathcal{AB}}$ can be implemented efficiently given sparse access to $H$.

\begin{definition}[Sparse access oracle~{\cite[Lemma~48]{gilyen2018QSingValTransf}}]\label{defn:sparse_access_oracle}
The matrix $H = [A,\,-\bvec]$ is accessed via three oracles:
\begin{itemize}
    \item \emph{Entry oracle} $\mathcal{O}_H$: $\ket{i}\ket{j}\ket{0^a} \mapsto \ket{i}\ket{j}\ket{h_{ij}}$, where $h_{ij}$ is an $a$-bit description of $H_{i,j}$.
    \item \emph{Row oracle} $\mathcal{O}_{\mathrm{r}}$: $\ket{i}\ket{k} \mapsto \ket{i}\ket{r_{ik}}$, where $r_{ik}$ is the column index of the $k$-th nonzero entry in row~$i$.
    \item \emph{Column oracle} $\mathcal{O}_{\mathrm{c}}$: $\ket{\ell}\ket{j} \mapsto \ket{c_{\ell j}}\ket{j}$, where $c_{\ell j}$ is the row index of the $\ell$-th nonzero entry in column~$j$.
\end{itemize}
\end{definition}

\begin{theorem}[Implementation cost]\label{thm:unitary-from-reflections}
Given the sparse access oracles of \cref{defn:sparse_access_oracle}, the unitary $U_{\mathcal{AB}} = (2\Pi_{\mathcal{A}} - I)(2\Pi_{\mathcal{B}} - I)$ can be implemented in time $\widetilde{O}(s)$.
\end{theorem}

\begin{proof}
Each row star state $\ket{\Psi_i}$ and each column star state $\ket{\Phi_j}$ (or $\ket{\Phi_b}$) has at most $s$ nonzero amplitudes, which can be computed using $O(s)$ queries to $\mathcal{O}_H$, $\mathcal{O}_{\mathrm{r}}$, and $\mathcal{O}_{\mathrm{c}}$.
A sparse state with $s$ nonzero entries can be prepared in time $\widetilde{O}(s)$~\cite{grover2002SuperposEffIntegrProbDistr,ramacciotti2024simple,luo2024circuit}.
By~\cite[Claims~2.10 and~2.11]{jeffery2025multidimensional}, the reflections $2\Pi_{\mathcal{A}} - I$ and $2\Pi_{\mathcal{B}} - I$ can each be implemented with cost asymptotically equal to that of preparing the corresponding star states.
Hence the total cost of $U_{\mathcal{AB}}$ is $\widetilde{O}(s)$.
\end{proof}

\subsection{The Algorithm and Its Complexity}\label{sec:algorithm}

We now present the complete algorithm.
The key idea is that phase estimation on $(U_{\mathcal{AB}},\,\ket{\Phi_b})$ produces a state $\ket{\theta'}$ close to $\ket{\theta^*}$, and a subsequent projective measurement extracts the solution~$\ket{\xvec}$.

\begin{algorithm}[ht]
\caption{Quantum linear system solver via quantum walk}\label{alg:newQLS2}
\begin{algorithmic}[1]
\REQUIRE Sparse access oracles for $A$ and $\bvec$
(\cref{defn:sparse_access_oracle}), with $H = [A,\,-\bvec]$ being
$s$-sparse; target precision $\epsilon \in (0,1/2)$; a norm bound
$r \geq 1$ with $1 \leq \|\xvec\|_2 \leq r$; a bound $\Lambda \geq \ET$.
\ENSURE A quantum state $\ket{\xvec'}$ satisfying
$\|\ket{\xvec'}\!\bra{\xvec'} - \ket{\xvec}\!\bra{\xvec}\|_1 \leq \epsilon$.

\STATE \textbf{Construct the walk operator.}\;
Build the graph $G$ from $H = [A,\,-\bvec]$ (\cref{sec:graph}) and prepare
the initial state $\ket{\Phi_b}$.
Implement $U_{\mathcal{AB}}$ via \cref{thm:unitary-from-reflections}.

\STATE \textbf{Phase estimation.}\;
Run quantum phase estimation on $(U_{\mathcal{AB}},\,\ket{\Phi_b})$ with
precision $O\!\bigl(\epsilon^2 / \sqrt{s \cdot \Lambda}\bigr)$.
Post-select on phase value~$0$ to obtain $\ket{\theta'}$; this succeeds
with probability at least $1/(1+r^2)$.

\STATE \textbf{Solution extraction.}\;
Apply the projective measurement
$\{I - \ket{\Phi_b}\!\bra{\Phi_b},\;\ket{\Phi_b}\!\bra{\Phi_b}\}$ to
$\ket{\theta'}$; the first outcome occurs with probability at least
$1/2 - \epsilon/6$, and conditioned on it the resulting state is
$\ket{\xvec'}$.
If either post-selection fails, restart from Step~1; $O(r^2)$ repetitions
suffice in expectation.
\end{algorithmic}
\end{algorithm}

\begin{theorem}[Main complexity result]\label{thm:QLSalgorithmcomplexity}
Let $A \in \mathbb{R}^{M \times N}$, let
$\bvec \in \mathrm{range}(A)\setminus\{0\}$ with $H = [A,\,-\bvec]$ being
$s$-sparse, let $\xvec = A^{+}\bvec$, and let
\[
    \ket{\xvec}
    \;=\;
    \frac{1}{\|\xvec\|_2}\sum_{j=1}^{N} y_j\,\ket{\Phi_j}
\]
denote the encoding of the normalized solution in the orthonormal column
star states of \cref{sec:graph}.

Suppose we are given sparse access oracles for $A$ and $\bvec$
(\cref{defn:sparse_access_oracle}), a target precision
$\epsilon \in (0,1/2)$, a norm bound $r \geq 1$ such that
$1 \leq \|\xvec\|_2 \leq r$, and a bound $\Lambda \geq \ET$. With
probability at least $2/3$, \cref{alg:newQLS2} outputs a quantum state
$\ket{\xvec'}$ satisfying
\[
    \bigl\|\ket{\xvec'}\!\bra{\xvec'} - \ket{\xvec}\!\bra{\xvec}\bigr\|_1
    \;\leq\; \epsilon,
\]
with query complexity
$O\!\bigl(r^2\,\sqrt{s \cdot \Lambda}\,/\,\epsilon^2\bigr)$
and time complexity
$\widetilde{O}\!\bigl(r^2\,\sqrt{s^3 \cdot \Lambda}\,/\,\epsilon^2\bigr)$.

In particular, if $\|\xvec\|_2$ is known exactly, then rescaling
$\bvec \mapsto \bvec/\|\xvec\|_2$ reduces to the case $r = 1$: the query
complexity is $O\!\bigl(\sqrt{s \cdot \Lambda}\,/\,\epsilon^2\bigr)$ and
the time complexity is
$\widetilde{O}\!\bigl(\sqrt{s^3 \cdot \Lambda}\,/\,\epsilon^2\bigr)$ where $ET \leq \Lambda$ now is for the rescaled linear system $Ax=\frac{\bvec}{\|\xvec\|_2}$, 
Moreover, the success probability can be raised to $1 - \nu$ for any
$\nu \in (0,1)$ at multiplicative cost $O(\log(1/\nu))$ in both bounds.
\end{theorem}

\begin{proof}
We analyze each step of \cref{alg:newQLS2}.

\paragraph{Step 1: Preparation.}
By \cref{thm:unitary-from-reflections}, preparing $\ket{\Phi_b}$ and
implementing one application of $U_{\mathcal{AB}}$ costs
$\widetilde{O}(s)$ time.

\paragraph{Step 2: Phase estimation.}
Apply \cref{thm:statepreparation} with precision parameter
$\epsilon' = \epsilon/6$, using the bound $\Lambda \geq \ET$ to set the
phase resolution. The phase estimation uses
$O\!\bigl(\sqrt{s \cdot \Lambda}/\epsilon^2\bigr)$ queries to
controlled-$U_{\mathcal{AB}}$, the post-selection on phase value~$0$
succeeds with probability at least
\[
    p_{\mathrm{pe}}
    \;=\;
    \frac{1}{1 + \|\xvec\|_2^2}
    \;\geq\;
    \frac{1}{1 + r^2},
\]
and, conditioned on success, produces a state $\ket{\theta'}$ satisfying
\begin{equation}\label{eq:theta-prime-guarantee}
    \tfrac{1}{2}\bigl\|\ket{\theta'}\!\bra{\theta'}
    - \ket{\theta^*}\!\bra{\theta^*}\bigr\|_1
    \;\leq\; \frac{\epsilon}{6}.
\end{equation}

\paragraph{Step 3: Solution extraction.}
Let $\Pi_{\mathrm{sol}} = I - \ket{\Phi_b}\!\bra{\Phi_b}$.
When applied to the ideal state $\ket{\theta^*}$, the measurement
succeeds with probability
\[
    \gamma
    \;=\;
    \operatorname{Tr}\!\bigl[\Pi_{\mathrm{sol}}\,
    \ket{\theta^*}\!\bra{\theta^*}\bigr]
    \;=\;
    \frac{\|\xvec\|_2^2}{1 + \|\xvec\|_2^2}
    \;\geq\;
    \frac{1}{2},
\]
where the inequality uses the promise $\|\xvec\|_2 \geq 1$. The
post-measurement state is
\[
    \ket{\xvec}
    \;=\;
    \frac{\Pi_{\mathrm{sol}}\,\ket{\theta^*}}
    {\|\Pi_{\mathrm{sol}}\,\ket{\theta^*}\|}
    \;=\;
    \frac{1}{\|\xvec\|_2}\sum_{j=1}^{N} y_j\,\ket{\Phi_j},
\]
which encodes the normalized solution to $A\xvec = \bvec$ in the column
star state basis.

\smallskip
We now show that the approximate input $\ket{\theta'}$ yields a state
close to $\ket{\xvec}$. Write $\rho = \ket{\theta^*}\!\bra{\theta^*}$,
$\sigma = \ket{\theta'}\!\bra{\theta'}$, and let
$\gamma' = \operatorname{Tr}(\Pi_{\mathrm{sol}}\,\sigma)$ be the success
probability on $\ket{\theta'}$.

\emph{Success probability stability.}\;
By the data-processing inequality for trace
distance~\cite[Theorem~9.1]{nielsen2002QCQI} and
\cref{eq:theta-prime-guarantee},
\[
    |\gamma' - \gamma|
    \;\leq\;
    \tfrac{1}{2}\|\rho - \sigma\|_1
    \;\leq\;
    \frac{\epsilon}{6},
\]
so $\gamma' \geq \tfrac{1}{2} - \tfrac{\epsilon}{6} > \tfrac{5}{12}$,
using $\epsilon < 1/2$.

\emph{Post-selected state closeness.}\;
Let
$\widetilde{\rho} = \Pi_{\mathrm{sol}}\rho\,\Pi_{\mathrm{sol}}/\gamma$
and
$\widetilde{\sigma} = \Pi_{\mathrm{sol}}\sigma\,\Pi_{\mathrm{sol}}/\gamma'$.
Then
\[
    \widetilde{\rho} - \widetilde{\sigma}
    \;=\;
    \frac{\Pi_{\mathrm{sol}}(\rho - \sigma)\Pi_{\mathrm{sol}}}{\gamma}
    \;+\;
    \Bigl(\frac{1}{\gamma} - \frac{1}{\gamma'}\Bigr)\,
    \Pi_{\mathrm{sol}}\,\sigma\,\Pi_{\mathrm{sol}}.
\]
Taking the trace norm and using
$\|\Pi_{\mathrm{sol}} X \Pi_{\mathrm{sol}}\|_1 \leq \|X\|_1$ for the
first term and
$\|\Pi_{\mathrm{sol}}\,\sigma\,\Pi_{\mathrm{sol}}\|_1 = \gamma'$ for the
second,
\[
    \|\widetilde{\rho} - \widetilde{\sigma}\|_1
    \;\leq\;
    \frac{\|\rho - \sigma\|_1}{\gamma}
    \;+\;
    \frac{|\gamma' - \gamma|}{\gamma}
    \;\leq\;
    \frac{\epsilon/3 + \epsilon/6}{1/2}
    \;=\;
    \epsilon,
\]
so
$\|\ket{\xvec'}\!\bra{\xvec'} - \ket{\xvec}\!\bra{\xvec}\|_1 \leq \epsilon$
as claimed.

\paragraph{Combining the steps.}
A single execution of Steps~1--3 succeeds --- i.e., both post-selections
accept --- with probability
\[
    p_{\mathrm{pe}} \cdot \gamma'
    \;\geq\;
    \frac{5/12}{1 + r^2}
    \;\geq\;
    \frac{5}{24\,r^2},
\]
using $1 + r^2 \leq 2r^2$ for $r \geq 1$. Hence $O(r^2)$ independent
repetitions produce an accepting run with probability at least $2/3$,
and $O(r^2 \log(1/\nu))$ repetitions raise this to $1 - \nu$. Each
repetition uses $O\!\bigl(\sqrt{s \cdot \Lambda}/\epsilon^2\bigr)$
queries by Step~2, and each application of controlled-$U_{\mathcal{AB}}$
costs $\widetilde{O}(s)$ time by Step~1, so a repetition takes
$\widetilde{O}\!\bigl(\sqrt{s^3 \cdot \Lambda}/\epsilon^2\bigr)$ time.
The total query complexity is therefore
$O\!\bigl(r^2\,\sqrt{s \cdot \Lambda}\,/\,\epsilon^2\bigr)$ and the total
time complexity is
$\widetilde{O}\!\bigl(r^2\,\sqrt{s^3 \cdot \Lambda}\,/\,\epsilon^2\bigr)$.

Finally, if $\|\xvec\|_2$ is known exactly, rescale
$(A,\bvec) \mapsto (A,\,\bvec/\|\xvec\|_2)$: the matrix
$H = [A,\,-\bvec/\|\xvec\|_2]$ remains $s$-sparse and the rescaled
solution has unit norm, so the promise holds with $r = 1$ (with $\ET$,
and hence $\Lambda$, taken for the rescaled instance). The bounds
specialize to $O\!\bigl(\sqrt{s \cdot \Lambda}/\epsilon^2\bigr)$ queries
and $\widetilde{O}\!\bigl(\sqrt{s^3 \cdot \Lambda}/\epsilon^2\bigr)$
time.
\end{proof}

\begin{remark}[On the error metric]\label{rem:error-metric}
The error is stated in trace distance: $\frac{1}{2}\|\ket{\xvec'}\!\bra{\xvec'} - \ket{\xvec}\!\bra{\xvec}\|_1 \leq \epsilon$.
For pure states, trace distance and Euclidean distance on the state vectors are related by
$\frac{1}{2}\|\ket{\xvec'}\!\bra{\xvec'} - \ket{\xvec}\!\bra{\xvec}\|_1 = \sqrt{1 - |\langle \xvec' | \xvec \rangle|^2}$,
so $\epsilon$-closeness in trace distance implies $|\langle \xvec' | \xvec \rangle|^2 \geq 1 - \epsilon^2$.
This is the standard fidelity guarantee used throughout the quantum algorithms literature, consistent with HHL~\cite{HHL09}.
\end{remark}

\section{Illustrative Examples}\label{sec:examples}

We present two examples demonstrating that $\ET$ can be dramatically smaller than the condition number $\kappa(A)$.
The first (\cref{sec:diagonal-example}) is a diagonal system that cleanly separates our algorithm from truncated HHL variants.
The second (\cref{sec:welded-tree}) applies our algorithm to the incidence linear system of a welded tree graph, achieving $O(1/\epsilon^2)$ complexity on an instance where $\kappa(B) = \Omega(2^{n/2})$.

\subsection{Separation from Truncated HHL}\label{sec:diagonal-example}

We begin with a transparent example showing that our algorithm succeeds where truncated HHL necessarily fails, regardless of the truncation threshold.

\begin{example}\label{ex:diagonal}
Let $A = \operatorname{diag}(\underbrace{1,\dots,1}_{n-1},\;1/n) \in \mathbb{R}^{n \times n}$ and $\bvec = (0,\dots,0,\;1/n)^{\!\top}$.
The unique solution is $\xvec = (0,\dots,0,\;1)^{\!\top}$, supported entirely on the smallest singular value $\sigma_n = 1/n$.
The condition number is $\kappa(A) = n$.
\end{example}

\paragraph{Our algorithm.}
We have $\pvec = (AA^{\!\top})^{+}\bvec = (0,\dots,0,\;n)^{\!\top}$ and the squared row norm of $H = [A,\,-\bvec]$ for the $n$-th row is $d_n = (1/n)^2 + (1/n)^2 = 2/n^2$; for $i < n$, $p_i = 0$ so those rows do not contribute.
Thus
\[
    \ET \;=\; \sum_{i=1}^{n} p_i^2\,d_i \;=\; n^2 \cdot \frac{2}{n^2} \;=\; 2,
\]
and \cref{thm:QLSalgorithmcomplexity} yields complexity $\widetilde{O}(1/\epsilon^2)$, independent of $\kappa(A) = n$.

\paragraph{Truncated HHL.}
The truncated approach of~\cite{HHL09} selects a threshold $\tau > 0$ and inverts only the singular values $\sigma_i \geq \tau$.
If $\tau > 1/n$, the algorithm discards the entire support of $\bvec$ (which lies in the eigenspace of $\sigma_n = 1/n$) and returns a state close to $\ket{\mathbf{0}}$, an incorrect answer.
If $\tau \leq 1/n$, the algorithm must invert $\sigma_n$, incurring cost $\Omega(1/\tau) \geq n = \kappa(A)$.
No choice of $\tau$ achieves both correctness and sublinear dependence on $\kappa(A)$.

\begin{remark}\label{rem:diagonal-insight}
This example isolates the key structural advantage of our approach.
The parameter $\ET$ captures the \emph{joint alignment} of $\bvec$ with $A$: here, $\bvec$ is supported entirely in the ill-conditioned subspace, the worst case for truncation, but benign for our algorithm because $\pvec$ concentrates on a single coordinate, making $\ET = O(1)$.
More generally, $\ET$ is small whenever $\pvec = (AA^{\!\top})^{+}\bvec$ is concentrated (in the weighted sense) on rows of $H$ with small norm, regardless of how $\bvec$ distributes across the spectrum of $A$.
\end{remark}

\subsection{The Welded Tree Graph}\label{sec:welded-tree}

Our second example applies the algorithm to the incidence linear system of a welded tree graph, achieving constant complexity on an instance with an exponentially large condition number.
This serves as a \emph{separation witness} between $\ET$ and $\kappa(B)$: the gap is exponential on a specific instance.

\begin{remark}[Role of this example]\label{rem:welded-role}
The welded tree example demonstrates that $\ET$ can be exponentially smaller than $\kappa(B)$ on a structured instance; it does not claim an end-to-end computational speedup.
The value lies in establishing that our complexity parameter captures instance structure that the condition number misses entirely.
See \cref{rem:alternative-approaches} for a comparison with other approaches that prepare related states.
\end{remark}

\subsubsection{Setup}

\begin{definition}[Welded tree graph]\label{def:welded-tree}
A \emph{welded tree graph~\cite{childs2003ExpSpeedupQW}} of depth~$n$ consists of two
complete binary trees $T_L, T_R$ of depth~$n$, with
roots $u$ (entry) and $v$ (exit), whose $2^n$ leaves on
each side are connected by a random cycle that alternates between
$T_L$- and $T_R$-leaves.
\end{definition}

\begin{remark}[Random vs.\ hierarchical welding]\label{rem:random-welding} In this work, we consider a specific family of welded tree
instances. 
The exponential separation between $\ET$ and $\kappa(B)$ requires hierarchical welding (see \cref{fig:welded}), which creates an $O(1)$-edge bottleneck between large vertex sets (see \cref{prop:kappa-bound}).
Crucially, $\ET = \Theta(1)$ under \emph{any} welding, the parameter depends only on the graph structure and $\bvec_{u,v}$.
\end{remark}

We interpret each edge as having unit resistance and study the resulting electrical network, following~\cite[Section~4]{vishnoi2013lx}.
Let $B \in \mathbb{R}^{N \times M}$ be the vertex--edge incidence matrix.
The \emph{incidence linear system} is
\begin{equation}\label{eqn:incidence}
    B\,\xvec \;=\; \bvec_{u,v},
\end{equation}
where $\bvec_{u,v}$ has $+1$ at $u$, $-1$ at $v$, and $0$ elsewhere.
The minimum-norm solution $\xvec = B^{+}\bvec_{u,v}$ is the $u$--$v$ unit electrical flow~\cite[Lemma~6]{wang2017efficient}: it splits evenly from parent to children in $T_L$ and merges symmetrically in $T_R$, with flow magnitude $|x_e| = 2^{-k}$ on edges at depth~$k$ (see \cref{fig:welded}).

\begin{figure}[ht]
    \centering
    \includegraphics[width=0.65\textwidth]{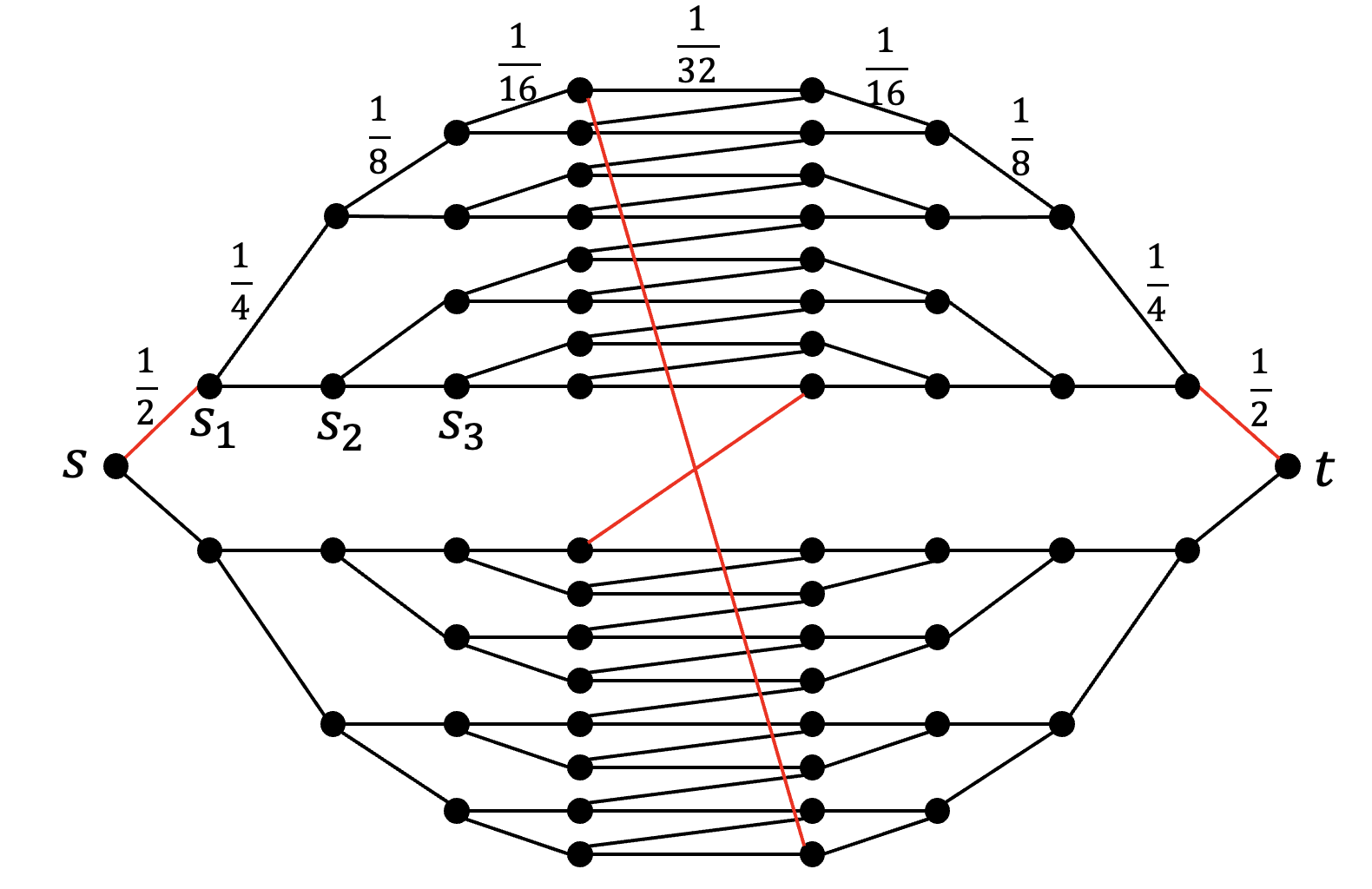}
    \caption{A welded tree graph of depth $n = 4$ with hierarchical welding.
    Edge labels indicate the $u$--$v$ electrical flow $\xvec = B^{+}\bvec_{u,v}$; sign is determined by the edge orientation in~$B$.
    The four red edges form a cut witnessing poor isoperimetric ratio: they separate a set of $\Theta(2^n)$ vertices with only $O(1)$ boundary edges.}
    \label{fig:welded}
\end{figure}

\subsubsection{Computing the Key Parameters}\label{sec:welded-params}

We compute $\|\xvec\|_2$, $\pvec$, and $\ET$ explicitly.
The vector $\pvec = (BB^{\!\top})^{+}\bvec_{u,v}$ is the minimum-$\ell_2$-norm vertex potential satisfying Kirchhoff's and Ohm's laws~\cite[Section~12]{spielman2019spectral},\cite[Section~2]{li2025multidimensional}, computed via the series and parallel laws for resistor networks~\cite[Section~2.3]{lyons2017probability}.
We emphasize that the following computations depend only on the tree structure and are \emph{independent of the welding}.

\begin{proposition}[Effective resistance]\label{prop:eff-resistance}
The effective resistance between $u$ and $v$ is
\[
    \mathcal{R}_{u,v}
    \;=\; \|\xvec\|_2^2
    \;=\; 2 - \tfrac{3}{2}\cdot 2^{-n} 
    \;=\; \Theta(1).
\]
\end{proposition}

\begin{proof}
At depth $k$ in either tree, the $2^k$ parallel unit-resistance edges have combined resistance $2^{-k}$.
These levels are in series from root to leaves, giving a total resistance $\sum_{k=1}^{n} 2^{-k} = 1 - 2^{-n}$ per tree.
The (middle) welding layer contributes $2^{-n}$ in parallel, and since the current splits across both sides of the welding, the effective contribution is $2^{-(n+1)}$.
The two trees plus welding are in series:
$\mathcal{R}_{u,v} = 2(1 - 2^{-n}) + 2^{-(n+1)} = 2 - \frac{3}{2}\cdot 2^{-n}$.
\end{proof}

\begin{proposition}[Potential vector]\label{prop:potential}
The minimum-$\ell_2$-norm potential $\pvec = (BB^{\!\top})^{+}\bvec_{u,v}$ is given by:
\begin{align*}
    &\text{At the roots:} & p_u &= 1 - \tfrac{3}{4}\cdot 2^{-n}, & p_v &= -(1 - \tfrac{3}{4}\cdot 2^{-n}). \\
    &\text{At depth $k$ in $T_L$:} & p_x &= 2^{-k} - \tfrac{3}{4}\cdot 2^{-n}. \\
    &\text{At depth $k$ in $T_R$:} & p_y &= -(2^{-k} - \tfrac{3}{4}\cdot 2^{-n}).
\end{align*}
This satisfies $\pvec \perp \mathbf{1}$, confirming it lies in $\operatorname{col}(B)$ and is the minimum-norm choice.
Moreover, $\|\pvec\|_2^2 = 4 - O(n \cdot 2^{-n}) = \Theta(1)$.
\end{proposition}

\begin{proof}
By Ohm's law, the potential drop across an edge at depth~$k$ equals the current $2^{-k}$ times the unit resistance.
The potential at a vertex at depth~$k$ from root~$u$ is $p_u$ minus the cumulative drop $\sum_{j=1}^{k} 2^{-j} = 1 - 2^{-k}$.
Setting $p_u = \mathcal{R}_{u,v}/2 = 1 - \frac{3}{4}\cdot 2^{-n}$ (by the network's symmetry about the welding) gives $p_x = 2^{-k} - \frac{3}{4}\cdot 2^{-n}$.
The anti-symmetry $p_x = -p_y$ for corresponding vertices ensures $\sum_x p_x = 0$, so $\pvec \perp \mathbf{1} = \ker(B^{\!\top})$.

For the norm, group vertices by depth.
At depth $k$ on each side there are $2^k$ vertices, each contributing $p_x^2 = (2^{-k} - \frac{3}{4}\cdot 2^{-n})^2$:
\begin{align*}
    \|\pvec\|_2^2
    &= 2\sum_{k=0}^{n} 2^k \bigl(2^{-k} - \tfrac{3}{4}\cdot 2^{-n}\bigr)^2 \\[2pt]
    &= 2\!\left[\sum_{k=0}^{n} 2^{-k} \;-\; \tfrac{3}{2}\cdot 2^{-n}\!\sum_{k=0}^{n}1 \;+\; \tfrac{9}{16}\cdot 4^{-n}\!\sum_{k=0}^{n} 2^k\right] \\[2pt]
    &= 2\bigl[\,(2 - 2^{-n}) - \tfrac{3}{2}\cdot 2^{-n}(n{+}1) + O(2^{-n})\,\bigr] \\[2pt]
    &= 4 - O(n\cdot 2^{-n}).
    \qedhere
\end{align*}
\end{proof}

\begin{proposition}[$\ET$ is constant]\label{prop:ET-constant}
For the incidence linear system~\eqref{eqn:incidence} on any welded tree graph of depth~$n$,
\[
    \ET \;=\; \sum_{x \in V} p_x^2\, d_x \;=\; \Theta(1).
\]
\end{proposition}

\begin{proof}
The augmented matrix is $H = [B,\,-\bvec_{u,v}]$.
Since the maximum degree is~$3$ and $\bvec_{u,v}$ has two nonzero entries, each row of $H$ has at most $4$ entries of magnitude~$1$.
Hence $d_x = \|H_{x,\cdot}\|_2^2 \leq 4$ for every vertex~$x$.

The crucial point is that $p_x^2$ \emph{decreases geometrically} with the depth of~$x$, overwhelming the exponential growth in the number of vertices at that depth.
The depth-$k$ contribution (from both trees) is
\begin{equation}\label{eq:ET-layer}
    \underbrace{2 \cdot 2^k}_{\text{vertices}} \;\times\;
    \underbrace{\bigl(2^{-k} - \tfrac{3}{4}\cdot 2^{-n}\bigr)^2}_{p_x^2} \;\times\;
    \underbrace{O(1)}_{d_x}
    \;\;=\;\;
    O\bigl(2^{-k}\bigr).
\end{equation}
Summing over $k = 0, 1, \ldots, n$: $\;\ET \leq 4\,\|\pvec\|_2^2 = \Theta(1)$.
\end{proof}

\begin{remark}[Why $\ET = \Theta(1)$ despite $|V| = \Theta(2^n)$]\label{rem:ET-convergence}
One might expect that summing over exponentially many vertices produces an exponentially large $\ET$.
The resolution is that the potential $p_x = \Theta(2^{-k})$ at depth~$k$ satisfies $p_x^2 = \Theta(4^{-k})$, which decays faster than the $2^k$ vertex count grows.
The depth-$k$ contribution is $\Theta(2^k \cdot 4^{-k}) = \Theta(2^{-k})$, a geometric series dominated by the $O(1)$ vertices near the roots.
\end{remark}

\subsubsection{Complexity of Our Algorithm}

Since $\bvec_{u,v}$ has two nonzero entries and the maximum degree is~$3$, the sparsity is $s = O(1)$.
By \cref{thm:QLSalgorithmcomplexity} with $\ET = \Theta(1)$, $s = O(1)$, and $\|\xvec\|_2 = \Theta(1)$:

\begin{corollary}\label{cor:welded-tree-complexity}
Our QLS algorithm solves the incidence linear system~\eqref{eqn:incidence} on the welded tree graph in time $\widetilde{O}(1/\epsilon^2)$, producing a quantum state $\epsilon$-close to the $u$--$v$ electrical flow state~$\ket{\xvec}$.
\end{corollary}

\subsubsection{Exponential Condition Number for Hierarchical Welding}\label{sec:kappa-lower-bound}

We now show that the hierarchical welding produces a graph with $\kappa(B) = \Omega(2^{n/2})$, establishing an exponential gap with $\ET = \Theta(1)$.
The key observation is that the hierarchical welding creates an $O(1)$-edge cut separating $\Theta(2^n)$ vertices.

\begin{proposition}\label{prop:kappa-bound}
For the welded tree graph with hierarchical welding, as in \cref{fig:welded},
$\kappa(B) = \Omega(2^{n/2})$.
\end{proposition}

\begin{proof}
Since $L = BB^{\!\top}$ and $\kappa(B) = \sqrt{\kappa(L)} = \sqrt{\|L\| \cdot \|L^{+}\|}$, we bound each factor separately.

\medskip\noindent
\emph{Lower bound on $\|L\|$:}\;
For any vertex~$i$ of degree~$\Delta_{\max} = 3$,
the Rayleigh quotient gives
$\|L\| \ge e_i^{\top} L\, e_i = \Delta_{\max} = 3$.

\medskip\noindent
\emph{Lower bound on $\|L^{+}\|$ via the isoperimetric ratio:}\;
By~\cite[Theorem~20.1.1]{spielman2019spectral}, $\|L^{+}\| \geq 2/\theta_G$, where $\theta_G = \min_{|S| \leq |V|/2} |\partial(S)|/|S|$ is the isoperimetric ratio.
We exhibit a set with exponentially small ratio.

Define $S$ to be the union of:
\begin{itemize}
    \item all vertices in the \emph{left subtree} of $T_L$ (the subtree rooted at $u$'s left child, containing $2^n - 1$ vertices at depths $1$ through $n$), and
    \item all vertices in the \emph{left subtree} of $T_R$ (the subtree rooted at $v$'s left child, containing $2^n - 1$ vertices).
\end{itemize}
Then $|S| = 2(2^n - 1) = 2^{n+1} - 2$.
Since $|V| = 2(2^{n+1} - 1) = 2^{n+2} - 2$, we have $|S| = 2^{n+1} - 2 \leq (2^{n+2}-2)/2 = 2^{n+1} - 1$, so $|S| \leq |V|/2$.

The edge boundary $\partial(S)$ consists of edges with exactly one endpoint in~$S$.
Within each binary tree, the left and right subtrees are connected \emph{only through the root}, so:
\begin{itemize}
    \item The edge from $u$'s left child to $u$ contributes $1$ edge ($u \notin S$).
    \item The edge from $v$'s left child to $v$ contributes $1$ edge ($v \notin S$).
    \item By the hierarchical welding, the left subtree's leaves of $T_L$ are matched exclusively with the left subtree's leaves of $T_R$, except for two vertices, so \textbf{only two welding edges leave} $S$.
\end{itemize}
Therefore $|\partial(S)| = 4$, giving
\[
    \theta_G
    \;\leq\;
    \frac{|\partial(S)|}{|S|}
    \;=\;
    \frac{4}{2^{n+1} - 2}
    \;=\; O(2^{-n}).
\]

Combining: $\|L^{+}\| \geq 2/\theta_G = \Omega(2^n)$, and hence
\[
    \kappa(B)
    \;=\; \sqrt{\|L\| \cdot \|L^{+}\|}
    \;\geq\; \sqrt{1 \cdot \Omega(2^n)}
    \;=\; \Omega(2^{n/2}).
    \qedhere
\]
\end{proof}

\begin{remark}[Dependence on the welding]\label{rem:welding-kappa}
The exponential condition number relies crucially on the hierarchical welding, which creates the $4$-edge cut in the proof above.
The hierarchical welding is chosen to maximize the gap between $\ET$ (which is $\Theta(1)$ regardless of welding) and $\kappa(B)$ (which becomes exponential).
\end{remark}

\subsubsection{Comparison with Prior Approaches}\label{sec:welded-comparison}

\paragraph{Standard QLS via HHL/QSVT.}
These algorithms solve $B\xvec = \bvec_{u,v}$ with complexity $\widetilde{O}(\kappa(B)\log(1/\epsilon))$.
For the hierarchical welding, $\kappa(B) = \Omega(2^{n/2})$, so the runtime is exponential in~$n$.

\paragraph{Truncated HHL.}
The truncated approach of~\cite{HHL09} inverts only the
component of $\bvec_{u,v}$ in the well-conditioned
eigenspaces of $B$. For the specific welded tree instance
analyzed here, $\bvec_{u,v}$ and the solution both lie in
an $O(n)$-dimensional invariant subspace of~$B$; however,
the analysis of truncated HHL would require complete
knowledge of the overlap between $\bvec_{u,v}$ and the
restricted eigenspaces of~$B$, which makes the viability
of truncated HHL on this instance unclear.
The welded tree example therefore serves as a separation
between $\mathit{ET}$ and $\kappa(A)$, not as a claim
of advantage over truncated HHL.

\paragraph{Prior quantum walk approaches.}
The electrical-network-based quantum walk algorithms of~\cite{belovs2013ElectricWalks,piddock2019electricfind,apers2022elfs} incorporate $\bvec$ through the potential vector, but their analysis requires the normalization $p_v = 0$ (the target vertex is grounded).
This amounts to using the shifted potential $\pvec' = \pvec + c\mathbf{1}$ with $c = -p_v = 1 - \frac{3}{4}\cdot 2^{-n} \approx 1$.
Since $\|\pvec'\|_2^2 = \|\pvec\|_2^2 + c^2|V| = \Theta(2^n)$, the shifted potential gives exponential complexity.

Our algorithm instead uses the minimum-$\ell_2$-norm potential ($\pvec \perp \mathbf{1}$), which has $\|\pvec\|_2^2 = \Theta(1)$.
This is possible because we \emph{know both endpoints} $u$ and $v$ and solve the incidence linear system directly, whereas the quantum walk framework is designed for \emph{search problems} where the target vertex is unknown.

\begin{remark}[Alternative state preparation]\label{rem:alternative-approaches}
The $u$--$v$ electrical flow state can also be prepared efficiently via the multidimensional electrical network framework~\cite{li2025multidimensional}, which modifies the linear system rather than solving $B\xvec = \bvec_{u,v}$ directly; the resulting quantum state encodes a different (but related) flow.
\end{remark}

\subsubsection{Discussion}

The welded tree graph is a canonical object in quantum computing, central to the oracle separation of~\cite{childs2003ExpSpeedupQW}.
Our result provides a new perspective: the $u$--$v$ electrical flow state can be prepared in $O(1/\epsilon^2)$ time by solving the incidence linear system with our instance-aware QLS algorithm, providing concrete evidence that $\ET$-based complexity captures the structural properties of the linear system.

A natural follow-up question is whether the efficient preparation of the electrical flow state can be leveraged for computational tasks on the welded tree graph.
Finding a $u$--$v$ path in the welded tree remains a major open problem in quantum query complexity~\cite{aaronson2021open}, and exponential quantum speedups for pathfinding have been demonstrated on related graph families~\cite{li2025exponential,li2025multidimensional,li2024exponential}.
The electrical flow is known to concentrate on short paths in electrical networks, suggesting a potential connection; we leave this exploration for future work.

\newcommand{\mlmon}{m}
\newcommand{\mlmonb}{m'}
\newcommand{\svec}{\vec{y}}
\newcommand{\tvec}{\vec{z}}

\section{Application to Multivariate Polynomial Systems}\label{sec:newalgpolysys}

We now demonstrate how our QLS algorithm, combined with a new right-rescaling technique, yields a quantum algorithm for solving systems of multivariate polynomial equations.
The starting point is the reduction of~\cite{chen2022quantum,ding2023limitations}: a Boolean polynomial system with $n$ variables and a unique solution of Hamming weight~$h$ can be encoded as a $\mathrm{poly}(n)$-sparse linear system of exponential dimension, the \emph{Boolean Macaulay linear system} $A\xvec = \bvec$, for which sparse access oracles can be computed implicitly.

The principal obstruction to a direct QLS approach is the tight lower bound $\kappa_{\bvec}(A) = \Omega(2^{h/2})$ established by~\cite{ding2023limitations}, which applies to all QLS algorithms (including ours) and to classical left-preconditioning.
We overcome this barrier by introducing a \emph{right-rescaling} that transforms the system into $AD\zvec = \bvec$, where $D$ is a diagonal matrix chosen to reduce $\|\zvec\|$ from $\Omega(2^{h/2})$ to $\sqrt{h}$.
Although $\kappa(AD)$ remains exponentially large (\cref{prop:weighted_condition_number}), our algorithm's complexity depends on $\ET$, not on the condition number. We show that $\ET$ can be polynomial for certain structured polynomial systems.

The section is organized as follows.
\cref{subsec:BooleanLinear} recalls the Boolean Macaulay formulation.
\cref{subsec:weightBoolinear} introduces the right-rescaling scheme and proves that the Boolean solution can still be efficiently extracted.
\cref{subsec:polysysalgorithm} presents the algorithm and its complexity in terms of~$\ET$.
\cref{subsec:polyexample} gives a worked example where $\ET = \mathrm{poly}(n)$.
\cref{subsec:MISpolysys} analyzes the planted maximum independent set problem and identifies a sufficient condition on the independent-set count under which the algorithm runs in polynomial time.

\subsection{Boolean Macaulay Linear System}\label{subsec:BooleanLinear}

We restate the formulation of~\cite{ding2023limitations}. \noindent
Throughout this section, we assume the system has a \emph{unique} Boolean solution of Hamming weight~$h$.
If multiple solutions exist, the Valiant--Vazirani affine-hashing technique~\cite{valiant1986NPEasyAsDetectingUniqueSols} (as applied in~\cite[Section~3]{ding2023limitations}) isolates one with high probability.

\begin{problem}[Boolean Multivariate Quadratic System over $\mathbb{C}$]\label{prob:MQC}
\textbf{Input:} Quadratic polynomials $\mathcal{F} = \{f_1, \dots, f_m\} \subseteq \mathbb{C}[x_1, \dots, x_n]$.\\
\textbf{Output:} A vector $\bs \in \{0,1\}^n$ satisfying $f_1(\bs) = \cdots = f_m(\bs) = 0$, if one exists.
\end{problem}

\begin{lemma}[Realification]\label{lem:realification}
Let $\mathcal{F} = \{f_1, \ldots, f_m\} \subseteq \mathbb{C}[x_1, \ldots, x_n]$
with $\deg f_j \leq 2$, and write $f_j = g_j + i\,h_j$ where
$g_j, h_j \in \mathbb{R}[x_1, \ldots, x_n]$ collect the real and imaginary
parts of the coefficients of $f_j$. Set
$\mathcal{F}_{\mathbb{R}} = \{g_1, h_1, \ldots, g_m, h_m\} \setminus \{0\}$.
Then $\deg p \leq 2$ for every $p \in \mathcal{F}_{\mathbb{R}}$,
$|\mathcal{F}_{\mathbb{R}}| \leq 2m$, and for every
$\bs \in \{0,1\}^n$,
\[
    f_j(\bs) = 0 \;\text{ for all } j
    \qquad\Longleftrightarrow\qquad
    p(\bs) = 0 \;\text{ for all } p \in \mathcal{F}_{\mathbb{R}}.
\]
In particular, $\mathcal{F}$ and $\mathcal{F}_{\mathbb{R}}$ have identical
Boolean solution sets, and a unique Boolean solution of one is a unique
Boolean solution of the other. 
\end{lemma}

\begin{proof}
Boolean points are real, and for real $\bs$ we have
$f_j(\bs) = g_j(\bs) + i\,h_j(\bs)$ with $g_j(\bs), h_j(\bs) \in \mathbb{R}$,
so $f_j(\bs) = 0$ if and only if $g_j(\bs) = h_j(\bs) = 0$.
\end{proof}

\noindent
By \cref{lem:realification} we assume henceforth, without loss of
generality, that $\mathcal{F} \subseteq \mathbb{R}[x_1, \ldots, x_n]$; all
matrices in this section are real, matching the real formulation of
\cref{sec:QLSalgorithm}.

\begin{definition}[Boolean Macaulay matrix~{\cite{ding2023limitations}}]\label{def:BooleanMacaulay}
The \emph{Boolean Macaulay matrix} $\bmacmat$ of degree~$d$ for $\mathcal{F} = \{f_1, \ldots, f_m\}$ is defined as follows.
Rows are indexed by pairs $(\mlmon, f)$ where $f \in \mathcal{F}$ and $\mlmon$ is a multilinear monomial such that $\deg(\psi(\mlmon f)) \leq d$.
The entry in column~$\mlmonb$ is the coefficient of $\mlmonb$ in $\psi(\mlmon f)$.
Columns span all multilinear monomials of degree at most~$d$.
Setting $d = n$ gives $\bmacmat \in \mathbb{C}^{(m \cdot 2^n) \times 2^n}$.
\end{definition}

\begin{definition}[Boolean Macaulay linear system]\label{def:booleanMacaulayLS}
Write $\bmacmat = [A,\,-\bvec]$, where $-\bvec$ is the column of constant terms.
The system $A\xvec = \bvec$ is the \emph{Boolean Macaulay linear system}.
\end{definition}

This formulation has two key properties that make it amenable to quantum algorithms:
\begin{enumerate}[label=(\roman*)]
    \item \textbf{Structured sparsity}~\cite[Lemma~5.4]{ding2023limitations}: $\bmacmat$ is $\mathrm{poly}(n)$-sparse, and the locations and values of nonzero entries can be computed implicitly, fitting the sparse-access model of QLS algorithms.

    \item \textbf{Efficient solution extraction}~\cite[Theorem~6.1]{ding2023limitations}: Given copies of the quantum state encoding the solution, a satisfying Boolean assignment can be recovered in $\mathrm{poly}(n)$ time.
\end{enumerate}

The obstacle is the exponential solution norm: the unique solution $\xvec = A^{+}\bvec$ satisfies $\|\xvec\|_2 = \sqrt{2^h - 1}$, since $\xvec$ encodes all $2^h - 1$ nonempty subsets of the support of~$\bs$.
This yields the tight lower bound $\kappa_{\bvec}(A) = \Theta(\|\xvec\|_2) = \Omega(2^{h/2})$~\cite[Corollary~5.6]{ding2023limitations}, which blocks all previous QLS approaches.

\subsection{Right-Rescaling: The Weighted Boolean Macaulay System}\label{subsec:weightBoolinear}

\subsubsection{The Rescaling Scheme}

Classical left-preconditioning replaces $A\xvec = \bvec$ with $DA\xvec = D\bvec$ for some invertible~$D$, preserving the solution~$\xvec$.
Since $\kappa_{\bvec}(A) = \Theta(\|\xvec\|_2)$ and the solution is unchanged, this cannot bypass the lower bound.

Our key idea is to \emph{right-multiply} instead:
\begin{equation}\label{eq:right-rescale}
    AD\,\zvec = \bvec,
    \qquad\text{where}\qquad
    \xvec = D\zvec.
\end{equation}
This changes the solution from $\xvec$ to $\zvec = D^{-1}\xvec$, and a judicious choice of $D$ can dramatically reduce $\|\zvec\|_2$.

\begin{remark}[Relationship to~\cite{li2025multidimensional}]\label{rem:rescaling-LZ}
The right-rescaling strategy is conceptually analogous to the ``alternative neighborhood'' technique of~\cite{li2025multidimensional}, which modified a graph-structured linear system to adjust the solution norm.
Both serve the same purpose, changing the norm of the solution relative to the system to facilitate efficient quantum state preparation, but operate in different access models: \cite{li2025multidimensional} added auxiliary rows (having only oracle access to the graph), whereas we right-multiply by a diagonal matrix (having implicit access to the Macaulay matrix).
\end{remark}

\begin{definition}[Weighted Boolean Macaulay linear system]\label{def:weighteddiagonal}
Let $D \in \mathbb{R}^{(2^n - 1) \times (2^n - 1)}$ be the diagonal matrix defined by
\[
    D_{jj} \;=\; \sqrt{\binom{h}{k}}
    \quad\text{if column } j \text{ corresponds to a monomial of degree } k \leq h,
    \qquad
    D_{jj} = 1 \quad\text{if } k > h.
\]
The \emph{weighted Boolean Macaulay linear system} is $AD\zvec = \bvec$.
\end{definition}

\noindent
The key effect of this rescaling is the following.

\begin{proposition}[Solution norm reduction]\label{prop:norm-reduction}
If the unique Boolean solution $\bs$ has Hamming weight~$h$ with support $S = \{x_j : \bs_j = 1\}$, then the solution $\zvec = D^{-1}\xvec$ of the weighted system satisfies $\|\zvec\|_2^2 = h$.
\end{proposition}

\begin{proof}
The nonzero entries of $\xvec$ are indexed by the $2^h - 1$ nonempty subsets of $S$, each equal to~$1$.
At degree~$k$, there are $\binom{h}{k}$ such subsets, each scaled by $1/\sqrt{\binom{h}{k}}$ in $\zvec$.
Hence $\|\zvec\|_2^2 = \sum_{k=1}^{h} \binom{h}{k} / \binom{h}{k} = h$.
\end{proof}

\noindent
Thus, the rescaling reduces the solution norm from $\sqrt{2^h - 1}$ to $\sqrt{h}$, an exponential reduction that circumvents the lower bound of~\cite{ding2023limitations}.

\begin{remark}[Requirements on $D$]\label{rem:D-requirements}
The specific choice in \cref{def:weighteddiagonal} is not the only valid option.
Any $D$ satisfying the following conditions is admissible:
(i)~$D$ is sparse and invertible;
(ii)~$AD$ remains $\mathrm{poly}(n)$-sparse;
(iii)~$\|\zvec\|_2^2 = \mathrm{poly}(n)$.
The choice above is tuned to make $\|\zvec\|_2^2 = h$ and to yield clean $\ET$ bounds for the MIS application (\cref{subsec:MISpolysys}).
\end{remark}

\subsubsection{Efficient Solution Extraction}

We now show that the Boolean solution can be recovered efficiently from measurements of $\ket{\zvec}$.

\begin{proposition}[Structure of $\ket{\zvec}$]\label{prop:zvec-structure}
Let $S_k = \binom{S}{k}$ denote the set of all $k$-element subsets of~$S$.
The normalized state $\ket{\zvec}$ decomposes as
\[
    \ket{\zvec} \;=\; \frac{1}{\sqrt{h}} \sum_{k=1}^{h} \ket{H_k},
    \qquad\text{where}\qquad
    \ket{H_k} \;=\; \frac{1}{\sqrt{\binom{h}{k}}} \sum_{T \in S_k} \ket{T}.
\]
\end{proposition}

\begin{lemma}[Solution recovery]\label{lem:solutionextractionlogn}
Given copies of $\ket{\zvec}$, the full support $S$ can be recovered with probability at least $1 - \delta$ using $O(\log(n/\delta))$ measurements.
\end{lemma}

\begin{proof}
Each measurement of $\ket{\zvec}$ first selects a degree~$k \in \{1,\ldots,h\}$ uniformly at random (probability $1/h$ each), then returns a uniformly random $k$-subset $T \subseteq S$.
For any fixed $x_j \in S$:
\[
    \Pr[x_j \in T]
    \;=\;
    \sum_{k=1}^{h} \frac{1}{h} \cdot \frac{k}{h}
    \;=\;
    \frac{h+1}{2h}
    \;>\; \frac{1}{2}.
\]
Hence $x_j$ is observed in each sample with constant probability exceeding $1/2$.
The probability that some $x_j \in S$ is missed in all $t$ samples is at most $h \cdot (1 - \frac{1}{2})^t = h \cdot 2^{-t}$ by the union bound.
Setting $t = \lceil \log_2(n/\delta) \rceil$ ensures the failure probability is at most~$\delta$.
\end{proof}

\subsubsection{The Condition Number Remains Exponential}

The rescaling reduces the solution norm but does \emph{not} improve the condition number.

\begin{proposition}\label{prop:weighted_condition_number}
$\kappa(AD) = \Tilde{\Omega}(2^{h/2})$.
\end{proposition}

\begin{proof}
We may assume $\|A\| = \mathrm{poly}(n)$ (by rescaling the polynomials if necessary).
The diagonal entries of $D$ range from $d_{\min} = 1$ (monomials of degree $\geq h$) to $d_{\max} = \sqrt{\binom{h}{\lfloor h/2 \rfloor}}$.
For a unit vector $\vec{e}_j$ corresponding to a degree-$(\lfloor h/2 \rfloor)$ monomial, $\|AD\vec{e}_j\| = d_{\max}\|A\vec{e}_j\| = \Tilde{\Omega}(2^{h/2})$.
For a unit vector $\vec{e}_k$ corresponding to a degree-$h$ monomial, $\|AD\vec{e}_k\| = \|A\vec{e}_k\| = \mathrm{poly}(n)$.
When the polynomial system has a unique Boolean solution, the
corresponding Boolean Macaulay linear system $A\vec y=\vec b$ of total
degree $n$ also has a unique solution
\cite[Theorem~A.2]{ding2023limitations}. Consequently, $A$ has full
column rank. Since $D$ is invertible, $AD$ also has full column rank.
Therefore,
\[
    \kappa(AD)
    =
    \frac{\sigma_{\max}(AD)}{\sigma_{\min}(AD)}
    \geq
    \frac{\|AD\vec e_j\|_2}{\|AD\vec e_k\|_2}.
\]
\end{proof}

\noindent
This exponential condition number rules out all previous QLS algorithms.
Our algorithm, whose complexity depends on $\ET$ rather than $\kappa(AD)$, is the first to potentially circumvent this barrier.

\subsection{The Algorithm}\label{subsec:polysysalgorithm}

We now present the quantum algorithm for Boolean polynomial systems.
It follows the framework of~\cite[Algorithm~1]{ding2023limitations}, replacing the QLS subroutine with our instance-aware algorithm applied to the weighted system.

\begin{algorithm}[ht]
\caption{Quantum algorithm for Boolean polynomial systems}\label{alg:MQC}
\begin{algorithmic}[1]
\REQUIRE Quadratic polynomials $\mathcal{F} = \{f_1, \dots, f_m\} \subseteq \mathbb{C}[x_1, \dots, x_n]$; Hamming weight guess~$h$.
\ENSURE A solution $\bs \in \{0,1\}^n$ with $f_i(\bs) = 0$ for all~$i$, if one exists.

\STATE Construct the weighted Boolean Macaulay linear system $AD\zvec = \bvec$ of total degree~$n$.
\STATE Apply the QLS algorithm (\cref{alg:newQLS2}) to prepare a state $\epsilon$-close to $\ket{\zvec}$.
\STATE Measure $\ket{\zvec}$ to obtain a subset $T \subseteq S$; set all variables indexed by $T$ to~$1$.
\STATE Repeat Steps~2--3 for $O(\log(n/\delta))$ rounds; set remaining variables to~$0$.
\STATE Return the assignment $\bs \in \{0,1\}^n$.
\end{algorithmic}
\end{algorithm}

\begin{theorem}[Complexity]\label{thm:qlspolytime}
Let $\mathcal{F}$ have a unique Boolean solution $\bs$ of Hamming weight~$h$.
Then \cref{alg:MQC} recovers $\bs$ with probability at least $1 - \delta$ in time
\[
    \widetilde{O}\!\left(\sqrt{\ET}\cdot \log(1/\delta)\right),
\]
where $\ET = \sum_i p_i^2 d_i$ is computed with respect to the weighted system $AD\zvec = \bvec$.
\end{theorem}

\begin{proof}
By \cref{thm:QLSalgorithmcomplexity} applied to the weighted system, the QLS step runs in time $\widetilde{O}(\sqrt{\ET}/\epsilon^2)$, where we absorb the dependence on $\|\zvec\|_2^2 = h \leq n$ and $s = \mathrm{poly}(n)$ into the $\widetilde{O}$ notation.
Setting $\epsilon = 1/\mathrm{poly}(n)$, the prepared state is sufficiently close to $\ket{\zvec}$ that each element of $S$ appears in a measurement outcome with probability at least $\frac{h+1}{2h} - \epsilon > 1/2$.
By \cref{lem:solutionextractionlogn}, $O(\log(n/\delta))$ measurements suffice to recover $S$ with probability $\geq 1 - \delta$.
\end{proof}

\begin{remark}[Unknown Hamming weight]\label{rem:unknown-h}
The algorithm requires the Hamming weight~$h$ as input (to define~$D$).
Since $h \in \{1,\dots,n\}$, one can iterate over all guesses and verify each output against the original system, incurring a multiplicative overhead of at most~$n$.
\end{remark}

\begin{remark}[Multiple solutions]\label{rem:multiple-solutions}
For systems with multiple Boolean solutions, we can follow the iterative approach of~\cite{chen2022quantum}: after each measurement fixes some variables, update the polynomial system by substitution and repeat.
At most $n$ iterations reduce the polynomial system to a trivially solvable one, with overall complexity dominated by the largest $\sqrt{\ET}$ encountered.
\end{remark}

\subsection{Worked Example: A Simple Polynomial System}\label{subsec:polyexample}

We demonstrate the rescaling technique on a transparent example.
Although the underlying problem is classically trivial, it illustrates the mechanism by which rescaling reduces $\ET$ from exponential to polynomial, and it serves as a sanity check before the more involved MIS analysis.

\begin{example}\label{ex:sum-poly}
Consider the single polynomial $f = x_1 + x_2 + \cdots + x_n - n$.
The unique Boolean solution is $\bs = (1,\ldots,1)$ with Hamming weight $h = n$.
\end{example}

\paragraph{Without rescaling.}
The Boolean Macaulay system $A\xvec = \bvec$ has $\|\xvec\|_2 = \sqrt{2^n - 1}$.
All QLS algorithms require time $\Omega(2^{n/2})$.

\paragraph{With rescaling.}
The weighted system $AD\zvec = \bvec$ with $D$ from \cref{def:weighteddiagonal} (setting $h = n$) has $\|\zvec\|_2^2 = n$ by \cref{prop:norm-reduction}.
Here $\bvec = (n, 0, \ldots, 0)^{\!\top}$ (only the first entry is nonzero).

We now show $\ET = \mathrm{poly}(n)$ by computing $\pvec$ explicitly.

\begin{proposition}\label{prop:toy-ET}
For the system $f = x_1 + \cdots + x_n - n$ with the weighted Macaulay formulation,
\[
    \ET \;=\; \mathrm{poly}(n).
\]
\end{proposition}

\begin{proof}
Since there is a single polynomial $f$, the Boolean Macaulay matrix $A$ is $2^n \times 2^n$.
Each row is indexed by a pair $(\mlmon, f)$ where $\mlmon$ is a multilinear monomial.

We compute $\pvec$ satisfying $A^{\!\top}\pvec = D^{-1}\zvec =: \vec{\zeta}$, where
\[
    \vec{\zeta}
    \;=\;
    D^{-1}\zvec
    \;=\;
    \biggl(\,
        \underbrace{\frac{1}{\binom{n}{1}},\ldots}_{\binom{n}{1}},\;\;
        \ldots,\;\;
        \underbrace{\frac{1}{\binom{n}{k}},\ldots}_{\binom{n}{k}},\;\;
        \ldots,\;\;
        \underbrace{\frac{1}{\binom{n}{n}}}_{\binom{n}{n}}
    \biggr).
\]

\emph{Structure of $A$:}\;
A monomial $\mlmon$ of degree~$k$ indexes a column of $A$ with $k+1$ nonzero entries: one from the row $(\mlmon, f)$ (with coefficient $k - n$) and $k$ from rows $(\mlmon/x_j, f)$ for each $x_j \mid \mlmon$ (with coefficient $1$).

\emph{Computing $p_{\mlmon}$:}\;
Using the constraint $\langle \pvec, \bvec \rangle = \|\zvec\|_2^2 = n$ and $\bvec = (n,0,\ldots,0)^{\!\top}$, we obtain $p_1 = 1$.
For a degree-$k$ monomial $\mlmon$, the equation $A^{\!\top}\pvec = \vec{\zeta}$ gives the recurrence
\begin{equation}\label{eq:recurrence-toy}
    \sum_{x_j \mid \mlmon} p_{\mlmon/x_j} \;+\; (k - n)\,p_{\mlmon}
    \;=\;
    \frac{1}{\binom{n}{k}}.
\end{equation}
By the symmetry of $f$ (all variables play identical roles), $p_{\mlmon}$ depends only on the degree~$k$; write $p_k := p_{\mlmon}$ for $\deg(\mlmon) = k$.
The recurrence becomes $k\,p_{k-1} + (k-n)\,p_k = 1/\binom{n}{k}$.

\emph{Claim:}\; $p_k = 1/\binom{n}{k}$ for all $0 \leq k \leq n$.

\emph{Verification:}\;
Using the identity $\binom{n}{k-1} = \frac{k}{n-k+1}\binom{n}{k}$:
\[
    \frac{k}{\binom{n}{k-1}} + \frac{k-n}{\binom{n}{k}}
    \;=\;
    \frac{n-k+1}{\binom{n}{k}} + \frac{k-n}{\binom{n}{k}}
    \;=\;
    \frac{1}{\binom{n}{k}}.
    \qquad\checkmark
\]

\emph{Bounding $\ET$:}\;
The squared row norm of the weighted augmented matrix $H = [AD,\,-\bvec]$ for a row indexed by a degree-$k$ monomial satisfies $d_k = \widetilde{O}(\binom{n}{k})$.
At degree~$k$, there are $\binom{n}{k}$ monomials.
Hence:
\[
    \ET
    \;=\;
    \sum_{k=0}^{n} \binom{n}{k} \cdot \frac{1}{\binom{n}{k}^2} \cdot \widetilde{O}\!\left(\binom{n}{k}\right)
    \;=\;
    \sum_{k=0}^{n} \widetilde{O}(1)
    \;=\;
    \widetilde{O}(n)
    \;=\;
    \mathrm{poly}(n).
    \qedhere
\]
\end{proof}

\noindent
By \cref{thm:qlspolytime}, \cref{alg:MQC} recovers $\bs = (1,\ldots,1)$ in polynomial time.
In contrast, all previous QLS-based approaches require exponential time on this instance due to $\kappa_{\bvec}(A) = \Omega(2^{n/2})$.

\subsection{Application to Planted Maximum Independent Set}\label{subsec:MISpolysys}

We now analyze the more substantial application: the planted maximum independent set (MIS) problem.
The goal is to recover a unique planted independent set of size~$h$ in a graph $G = (V, E)$ on $n$ vertices.
This can be formulated as the following polynomial system~\cite{lovasz1994stable}:
\begin{align}
    x_i^2 - x_i &= 0, &\text{for every } i \in V, \label{eq:boolean} \\
    x_i x_j &= 0, &\text{for every } \{i,j\} \in E, \label{eq:edge} \\
    \textstyle\sum_{i=1}^{n} x_i &= h. & \label{eq:weight}
\end{align}

The constraints~\eqref{eq:boolean} enforce Boolean values, and~\eqref{eq:edge} ensures no two adjacent vertices are selected.
The constraint~\eqref{eq:weight} fixes the set size.

\begin{remark}[Scope]
The decision version of MIS is NP-complete~\cite{karp2009reducibility}.
We do not claim to solve the general problem; we study the \emph{planted} variant, where a unique MIS of size~$h$ is embedded and the goal is recovery.
\end{remark}

\subsubsection{Structure of the Weighted Boolean Macaulay System for MIS}

The Boolean constraints~\eqref{eq:boolean} are automatically enforced by the multilinear restriction in the Macaulay construction.
The edge constraints~\eqref{eq:edge} imply that if a multilinear monomial $\mlmon$ contains any $x_i x_j$ with $\{i,j\} \in E$, then $\mlmon$ evaluates to zero on every independent set.
Equivalently, the only monomials contributing nonzero columns in the Macaulay matrix are those whose variable set forms an independent set of~$G$.

Let $S = \{x_{i_1}, \ldots, x_{i_h}\}$ denote the planted MIS, and let $f = \sum_{i=1}^{n} x_i - h$.
The weighted system $AD\zvec = \bvec$ with $D$ from \cref{def:weighteddiagonal} has $\bvec = (h, 0, \ldots, 0)^{\!\top}$ and
\[
    \zvec \;=\;
    \biggl(\,
        \underbrace{\frac{1}{\sqrt{\binom{h}{1}}},\ldots}_{\binom{h}{1}},\;\;
        \ldots,\;\;
        \underbrace{\frac{1}{\sqrt{\binom{h}{k}}},\ldots}_{\binom{h}{k}},\;\;
        \ldots,\;\;
        \underbrace{\frac{1}{\sqrt{\binom{h}{h}}}}_{\binom{h}{h}},\;\;
        0,\ldots,0
    \biggr),
\]
where the nonzero entries correspond to subsets of $S$ (the planted set), and $\|\zvec\|_2^2 = h$.

The solution extraction proceeds exactly as in \cref{lem:solutionextractionlogn}: $O(\log(n/\delta))$ measurements of $\ket{\zvec}$ suffice to recover $S$.

\subsubsection{Computing $\pvec$ and Bounding $\ET$}

The vector $\pvec$ satisfies $A^{\!\top}\pvec = D^{-1}\zvec =: \vec{\zeta}$.
By the structure of~$f$, a degree-$k$ monomial $\mlmon$ appears in $k+1$ rows of $A$: one from $(\mlmon, f)$ and $k$ from $(\mlmon/x_j, f)$ for each $x_j \mid \mlmon$.
The recurrence for $p_{\mlmon}$ is:
\begin{equation}\label{eq:recurrence-MIS}
    \sum_{x_j \mid \mlmon} p_{\mlmon/x_j} \;+\; (k - h)\,p_{\mlmon}
    \;=\;
    \begin{cases}
        1/\binom{h}{k} & \text{if $\mlmon$ is supported on $S$}, \\
        0 & \text{otherwise}.
    \end{cases}
\end{equation}

Unlike the toy example, the solution depends on the \emph{overlap} between $\mlmon$ and the planted set~$S$.
Let $t = |\mathrm{supp}(\mlmon) \setminus S|$ denote the number of variables in $\mlmon$ outside~$S$.
Write $p_{k,t}$ for the coefficient $p_{\mlmon}$ when $\deg(\mlmon) = k$ and overlap is~$t$.
\begin{proposition}[Boundary values and monotonicity]\label{prop:boundary-p}
Let $p_{k,t}$ denote the coefficient $p_{\mlmon}$ for a degree-$k$ monomial $\mlmon$ with
$t = |\mathrm{supp}(\mlmon) \setminus S|$ variables outside the planted set.
Then:
\begin{enumerate}[label=\textup{(\alph*)}]
    \item $p_{k,0} = 1/\binom{h}{k}$ for all $0 \leq k \leq h$.
    \item $p_{k,k} = 1/\binom{h-1}{k}$ for all $0 \leq k \leq h-1$.
    \item For each fixed $k$, the map $t \mapsto p_{k,t}$ is nondecreasing on $\{0,1,\ldots,k\}$.
\end{enumerate}
In particular, $p_{k,t} \in \bigl[1/\binom{h}{k},\; 1/\binom{h-1}{k}\bigr]$ for all valid $(k,t)$.
\end{proposition}

\begin{proof}
We first derive the general recurrence.
A degree-$k$ monomial $\mlmon$ with overlap $t = |\mathrm{supp}(\mlmon) \setminus S|$ contains $k - t$ variables in $S$ and $t$ variables outside~$S$.
Removing a variable $x_j \in S$ produces a degree-$(k{-}1)$ monomial with overlap~$t$;
removing $x_j \notin S$ produces overlap $t - 1$.
The recurrence~\eqref{eq:recurrence-MIS} therefore specializes to
\begin{equation}\label{eq:general-recurrence}
    (k - t)\,p_{k-1,\,t} \;+\; t\,p_{k-1,\,t-1} \;+\; (k - h)\,p_{k,t}
    \;=\;
    \begin{cases}
        1/\binom{h}{k} & \text{if } t = 0, \\[2pt]
        0 & \text{if } t \geq 1,
    \end{cases}
\end{equation}
with the initial condition $p_{0,0} = 1$ (from $\langle \pvec, \bvec \rangle = h$ and $\bvec = (h,0,\ldots)^{\!\top}$).

\medskip
\noindent\textbf{Part (a): $\boldsymbol{p_{k,0} = 1/\binom{h}{k}}$.}\;
Setting $t = 0$ in~\eqref{eq:general-recurrence}:
\begin{equation}\label{eq:t0-recurrence}
    k\,p_{k-1,0} + (k - h)\,p_{k,0} \;=\; \frac{1}{\binom{h}{k}}.
\end{equation}
We proceed by induction.
The base case $p_{0,0} = 1 = 1/\binom{h}{0}$ is immediate.
For the inductive step, assume $p_{k-1,0} = 1/\binom{h}{k-1}$.
Solving~\eqref{eq:t0-recurrence}:
\[
    p_{k,0}
    \;=\;
    \frac{1}{h - k}\!\left(\frac{k}{\binom{h}{k-1}} - \frac{1}{\binom{h}{k}}\right).
\]
The identity $\binom{h}{k-1} = \frac{k}{h-k+1}\,\binom{h}{k}$ gives $k/\binom{h}{k-1} = (h{-}k{+}1)/\binom{h}{k}$.  Hence:
\[
    p_{k,0}
    \;=\;
    \frac{(h{-}k{+}1) - 1}{(h-k)\,\binom{h}{k}}
    \;=\;
    \frac{1}{\binom{h}{k}}.
    \qquad\checkmark
\]

\medskip
\noindent\textbf{Part (b): $\boldsymbol{p_{k,k} = 1/\binom{h-1}{k}}$.}\;
Setting $t = k$ in~\eqref{eq:general-recurrence} (all variables outside $S$), every sub-monomial has overlap $k - 1$:
\begin{equation}\label{eq:tk-recurrence}
    k\,p_{k-1,\,k-1} + (k - h)\,p_{k,k} \;=\; 0
    \qquad\Longrightarrow\qquad
    p_{k,k} \;=\; \frac{k}{h - k}\,p_{k-1,\,k-1}.
\end{equation}
Again by induction.
Base cases: $p_{0,0} = 1 = 1/\binom{h-1}{0}$ and $p_{1,1} = 1/(h{-}1) = 1/\binom{h-1}{1}$.
For the inductive step, assume $p_{k-1,k-1} = 1/\binom{h-1}{k-1}$.  Then:
\[
    p_{k,k}
    \;=\;
    \frac{k}{h-k}\cdot\frac{(k{-}1)!\,(h{-}k)!}{(h{-}1)!}
    \;=\;
    \frac{k!\,(h{-}1{-}k)!}{(h{-}1)!}
    \;=\;
    \frac{1}{\binom{h-1}{k}}.
    \qquad\checkmark
\]
This is valid for $k \leq h - 1$, since~\eqref{eq:tk-recurrence} requires $h - k > 0$.

\medskip
\noindent\textbf{Part (c): Monotonicity.}\;
Define $\Delta_{k,t} := p_{k,t+1} - p_{k,t}$ for $0 \leq t \leq k - 1$.
We prove $\Delta_{k,t} \geq 0$ for all valid $(k,t)$ by induction on~$k$.

\emph{Base case} ($k = 1$): $\Delta_{1,0} = \frac{1}{h-1} - \frac{1}{h} = \frac{1}{h(h-1)} > 0$.\;\checkmark

\emph{Inductive step}: Assume $\Delta_{k-1,t} \geq 0$ for all $0 \leq t \leq k - 2$.
We derive a recurrence for $\Delta_{k,t}$ by subtracting the instances of~\eqref{eq:general-recurrence} for $(k,\,t{+}1)$ and $(k,\,t)$.
Two cases arise depending on whether $t = 0$ (where the right-hand sides differ).

\smallskip
\emph{Case $t \geq 1$}:\;
Both instances have right-hand side $0$.  Subtracting:
\begin{align*}
    (h{-}k)\,\Delta_{k,t}
    &= \bigl[(k{-}t{-}1)\,p_{k-1,t+1} + (t{+}1)\,p_{k-1,t}\bigr]
       - \bigl[(k{-}t)\,p_{k-1,t} + t\,p_{k-1,t-1}\bigr] \\
    &= (k{-}t{-}1)\underbrace{\bigl(p_{k-1,t+1} - p_{k-1,t}\bigr)}_{\Delta_{k-1,t}}
       \;+\; t\underbrace{\bigl(p_{k-1,t} - p_{k-1,t-1}\bigr)}_{\Delta_{k-1,t-1}}.
\end{align*}
Since $k - t - 1 \geq 0$ (as $t \leq k{-}1$), $t \geq 1$, $h - k > 0$, and both $\Delta_{k-1,t},\Delta_{k-1,t-1} \geq 0$ by induction, we get $\Delta_{k,t} \geq 0$.

\smallskip
\emph{Case $t = 0$}:\;
The instance for $(k,1)$ has right-hand side $0$, while $(k,0)$ has $1/\binom{h}{k}$.  Subtracting:
\begin{align*}
    (h{-}k)\,\Delta_{k,0}
    &= \bigl[(k{-}1)\,p_{k-1,1} + p_{k-1,0}\bigr]
       - \bigl[k\,p_{k-1,0} - 1/\binom{h}{k}\bigr] \\
    &= (k{-}1)\underbrace{\bigl(p_{k-1,1} - p_{k-1,0}\bigr)}_{\Delta_{k-1,0}}
       \;+\; \frac{1}{\binom{h}{k}}.
\end{align*}
Since $k - 1 \geq 0$, $\Delta_{k-1,0} \geq 0$ by induction, and $1/\binom{h}{k} > 0$, we get $\Delta_{k,0} > 0$.

\medskip
This completes the induction.
Combined with parts~(a) and~(b), the nondecreasing property yields $p_{k,t} \in \bigl[p_{k,0},\,p_{k,k}\bigr] = \bigl[1/\binom{h}{k},\; 1/\binom{h-1}{k}\bigr]$.
\end{proof}

\noindent
To bound $\ET$, let $I_{k,t}$ denote the number of independent sets of size~$k$ containing exactly $t$ vertices outside~$S$.
The squared row norm for degree-$k$ rows satisfies $d_k = \widetilde{O}(\binom{h}{k})$.
Then:
\begin{equation}\label{eq:ET-MIS}
    \ET
    \;=\;
    1 + \sum_{k=1}^{h}\; \sum_{t=0}^{k}\; I_{k,t} \cdot p_{k,t}^2 \cdot d_k.
\end{equation}
Since $p_{k,t}^2 \in \bigl[1/\binom{h}{k}^2,\; 1/\binom{h-1}{k}^2\bigr]$ and $d_k = \widetilde{O}(\binom{h}{k})$, the product $p_{k,t}^2 \cdot d_k$ lies in the range $\widetilde{\Theta}(1/\binom{h}{k})$ to $\widetilde{\Theta}(1/\binom{h-1}{k})$.

\subsubsection{Sufficient Condition for Polynomial Time}

We can now state a clean sufficient condition.

\begin{theorem}[Polynomial-time condition for planted MIS]\label{thm:MIS-condition}
Let $G$ be a graph with a unique planted MIS of size~$h$, and let $I_k = \sum_{t=0}^{k} I_{k,t}$ denote the total number of independent sets of size~$k$ in $G$.
If
\begin{equation}\label{eq:MIS-condition}
    I_k \;\leq\; \mathrm{poly}(n) \cdot \binom{h}{k}
    \qquad\text{for all } 1 \leq k \leq h,
\end{equation}
then $\ET = \mathrm{poly}(n)$ and \cref{alg:MQC} recovers the planted set in polynomial time.
\end{theorem}

\begin{proof}
Under condition~\eqref{eq:MIS-condition}, each term in~\eqref{eq:ET-MIS} satisfies
\[
    \sum_{t=0}^{k} I_{k,t} \cdot p_{k,t}^2 \cdot d_k
    \;\leq\;
    I_k \cdot \frac{\widetilde{O}(\binom{h}{k})}{\binom{h-1}{k}^2}
    \;=\;
    \mathrm{poly}(n) \cdot \binom{h}{k} \cdot \frac{\widetilde{O}(\binom{h}{k})}{\binom{h-1}{k}^2}.
\]
Using $\binom{h-1}{k} = \frac{h-k}{h}\binom{h}{k}$:
\[
    \frac{\binom{h}{k}^2}{\binom{h-1}{k}^2}
    \;=\;
    \frac{h^2}{(h-k)^2}
    \;\leq\;
    h^2
    \qquad\text{(for $k \leq h-1$)}.
\]
Hence each of the $O(h)$ terms in the sum is $\mathrm{poly}(n)$, and $\ET = \mathrm{poly}(n)$.
The result follows from \cref{thm:qlspolytime}.
\end{proof}

\subsubsection{Discussion: The Condition and Classical Hardness}\label{sec:MIS-discussion}

We now discuss the condition $I_k \leq \mathrm{poly}(n) \cdot \binom{h}{k}$ and its relationship to computational hardness.

\paragraph{Interpretation.}
Condition~\eqref{eq:MIS-condition} asks that the number of independent sets of each size~$k$ is not much larger than $\binom{h}{k}$, the number of subsets of the planted set itself.
Informally, the graph should not have ``too many'' independent sets beyond the subsets of the planted MIS.
This is a structural property of the graph that interpolates between trivial instances (where the condition holds vacuously because the graph is very dense) and worst-case instances (where $I_k$ can grow combinatorially).

\paragraph{What we do and do not claim.}
We do \emph{not} claim that condition~\eqref{eq:MIS-condition} is satisfied by any specific family of classically hard graphs.
Identifying such a family, or proving that none exists, is an important open problem.
What we \emph{do} claim is that our framework reduces the planted MIS problem to an explicit combinatorial condition on independent-set counts, which can be analyzed on a case-by-case basis for specific graph families.
This is a qualitative improvement over prior QLS approaches, which are blocked by the unconditional $\Omega(2^{h/2})$ lower bound regardless of graph structure.

\paragraph{Instance dependence.}
A notable feature of our algorithm is that its runtime depends on the specific graph instance (through $\ET$), whereas known classical algorithms for planted MIS on random graphs, such as spectral methods and approximate message passing~\cite{gamarnik2021overlap}, are \emph{instance-insensitive}.
For certain random graph models, the overlap gap property~\cite{gamarnik2021overlap} is known to preclude local algorithms from finding the planted set.
Whether our instance-dependent quantum algorithm can succeed in regimes where classical algorithms are blocked by the overlap gap property remains an intriguing open question.

\newpage
\textbf{AI Disclosure.}
We used  ChatGPT and Claude to improve the clarity of the
writing, provide editorial feedback, and check the internal consistency of
the manuscript. The tool did not materially affect the scientific content,
technical arguments, results, or conclusions. The authors independently
reviewed and verified all content, including the references, and take full
responsibility for its correctness and originality.

\section*{Acknowledgments}
J.L. thanks Nai-Hui Chia for many helpful discussions and Yuan Su, Stephen Jordan for valuable feedback on the initial version of the paper. J.L. thanks Giacomo Nannicini and Sevag Gharibian for asking and discussing the BQP completeness of the QLS problem for sparse matrix $A$ and sparse vector $\vec{b}$, respectively. J.L. also thanks Kunal Marwaha for his helpful comments on the introduction of the initial version of the paper.  Part of the ideas about the rescaling scheme were developed when J.L. was a Ph.D. student at Penn State University. J.L. also thanks Sean Hallgren for many helpful discussions and conversations. 
J. L. also thanks András Gilyén for previous collaboration on the limitations of using existing QLS algorithms to solve polynomial systems, Sebastian Zur for the collaboration on multidimensional electrical networks with respect to linear systems related to graphs, and Yu Tong for the collaboration on the use of the optimal eigenstate filtering technique for the pathfinding problem that makes this result possible. J.L. was supported by NSF Award FET-2243659 and NSF Career Award FET-2339116, Welch Foundation Grant no. A22-0307, a Microsoft Research Award, an Amazon Research Award, and a Ken Kennedy Research Cluster award (Rice K2I), in part through seed funding from the Ken Kennedy Institute at Rice University, the Rice CS Department, and the George R. Brown School of Engineering and Computing, by Rice University and the Department of Computer Science at Rice University, by the Ken Kennedy Institute and Rice Quantum Initiative, which is part of the Smalley-Curl Institute.
    \bibliographystyle{alpha}
\bibliography{qlspbib,extra,GI,indep}

\appendix 

\section{New QLS Algorithm based on Kernel Projection} \label{append:QLSQSVT}

In this section, we present an alternative QLS algorithm, which is based on the \textit{ kernel projection} technique introduced in \cite{dalzell2024shortcut} and built upon the QSVT framework. 
Specifically, the Kernel Projection maps the initial state into the null space of the matrix
\[
H = [A, -\bvec],
\]
which is closely related to the solution of the linear system $ A\xvec = \bvec $, as outlined in \cref{thm:vectordecomposition}. 
A subsequent projective measurement then allows us to efficiently recover the solution of the linear system, providing a direct method for extracting the quantum state corresponding to the solution vector $ \xvec $.

\begin{algorithm}
\caption{New QLS Algorithm based on Kernel Projection}
\begin{algorithmic}[1]\label{alg:newQLS1}
\REQUIRE  Oracle access of sparse matrix $A$ and sparse vector $\bvec$.

\ENSURE A quantum state $\ket{\xvec}$ such that \[
\|\ket{\xvec'}-\ket{\xvec}\|_2 \le \epsilon.
\]
 
 \STATE Prepare the initial state $\ket{\psi_0}=\ket{e_{n+1}}$.
 \STATE Apply kernel projection to approximately reflect about the kernel of $H$, 
 \[
 V\ket{0}^{\otimes a}\ket{\psi_0} = \ket{0}^{\otimes a}  F_{\Delta,\ell}(H)\ket{\psi_0} + \ket{\perp}
 \]
measure the first register to obtain the quantum state $\ket{\theta'}$.
 \STATE By measuring the operator $I-\ket{e_{n+1}}\bra{e_{n+1}}$ on the resulting state $f(H)\ket{\psi_0}$,  obtain the quantum state $\ket{\xvec'}$. Repeat until success. 
 
\end{algorithmic}
\end{algorithm}

\begin{theorem}
   Using $O(\kappa(A)\log(1/\epsilon))$ queries of $U_{H}$ and $U_H^{+}$,  the output of \cref{alg:newQLS1} is a quantum state $\ket{\xvec'}$ that is $\epsilon$ close to the quantum state $\ket{\xvec}$, where $A\xvec=\bvec$.
\end{theorem}

\begin{proof}
    We analyze the complexity of \cref{alg:newQLS1} for each step:

\begin{enumerate}
    \item \textbf{Preparation of the initial state (Step 1):} \\
    The preparation of the initial state in Step 1 takes constant time since $ \ket{e_{n+1}} $ is a basis vector.

    Define the following quantum states:
    \[
    \ket{\theta} = \frac{1}{\sqrt{1+\|\xvec\|_2^2}}\left(\sum_{i=1}^{N}\theta_i\ket{e_i} + \ket{e_{n+1}}\right)
    \]
    \[
    \ket{\theta^{\perp}} = \frac{1}{\|\xvec\|_2 \sqrt{1+\|\xvec\|_2^2}}\left(\sum_{i=1}^{N}\theta_i\ket{e_i} - \|\xvec\|_2^2\ket{e_{n+1}}\right)
    \]

    By \cref{thm:vectordecomposition}, we have:
    \[
    \ket{\psi_0} = \frac{1}{\sqrt{1+\|\xvec\|_2^2}}\left( \ket{\theta} - \ket{\theta^{\perp}} \right)
    \]

    \item \textbf{Kernel Projection (Step 2):} \\
    By \cref{lem:kernelprojection} and \cref{thm:kappaequ}, using $ \Theta(\ell)=\Theta(\kappa(A)\log (1/\epsilon)) $ queries to $ U_H $ and $ U_H^{+} $, we can construct a unitary $ V $ such that:
    \[
    V\ket{0}^{\otimes a} \ket{\psi_0} = \ket{0}^{\otimes a} F_{\Delta,\ell} \ket{\psi_0} + \ket{\perp}
    \]
    With probability at least $ \frac{1}{1+\|\xvec\|_2^2} $, we obtain a quantum state $ \ket{\theta'} $ satisfying:
    \[
    \|\ket{\theta'} - \ket{\theta}\|_2 \leq \epsilon
    \]

    \item \textbf{Projective Measurement (Step 3):} \\
    As noted in the last step of Algorithm 1 in \cite{dalzell2024shortcut}, this projective measurement can be implemented using a single multi-controlled Toffoli gate. If the ancillary qubit is measured in state $\ket{0}$, with probability $ \frac{\|\xvec\|_2^2}{1+\|\xvec\|_2^2} $, we obtain the resulting state $ \ket{\xvec'}$ that approximates $ \ket{\xvec}$.
\end{enumerate}

When the norm $\|\xvec\|_2$ is equal to 1, the success probability of obtaining the desired state $\ket{\xvec'}$ in each iteration is at least $1/4$. Therefore, by repeating the above steps a constant number of times, specifically $O(1)$ iterations, we can ensure that we obtain a state $\epsilon$ close to $\ket{\xvec}$ with high probability.
\end{proof}

Quantum algorithms for sparse matrices based on the QSVT framework can be dequantized using classical algorithms inspired by quantum techniques, as shown in~\cite{gharibian2022dequantizing}. Since our QLS algorithm is implemented within the QSVT framework, it naturally admits a classical counterpart with complexity $\widetilde{O}(s^{d})$ \cite[Theorem 4.1]{gharibian2022dequantizing} or \cite[Theorem 8]{le2025robust}, where $d$ is the degree of the polynomial used under the QSVT framework. Note that the dequantization result for the QLS algorithm presented in~\cite{le2025robust} differs slightly from ours. Specifically, their method constructs a polynomial approximation to the inverse function $A^{-1}$ with degree $O\left(\kappa(A)\log\left(\kappa(A)/\epsilon\right)\right)$, adapted to the condition number $A$. In contrast, our QLS algorithm applies a minimax polynomial of degree $O\left(\kappa(A)\log(1/\epsilon)\right)$ (\cref{eq:minimaxpoly}) associated with the augmented matrix $H = [A, -\bvec]$, thus giving a distinct complexity.

  \section{Kernel Projection via QSVT} 

We restate the result of the kernel projection in \cite[Appendix B.2]{dalzell2024shortcut}, which serves as a key step in our alternative QLS algorithm.
 \begin{definition}[Block encoding {\cite[Definition 43]{gilyen2018QSingValTransf}}]
\label{defn:block_encoding}
Suppose that $H$ is a $s$-qubit operator and let $ \epsilon\in \mathbb{R}_{+}$ and $a\in \mathbb{N}$. We say that the $(s+a)$-qubit unitary $U$ is an $(1,a,\epsilon)$ block encoding of H on the registers $\beta_1$ and $\beta_2$, if
\[
\|H-\left(\bra{0}_{\beta_1}^{\otimes a}\otimes I_{\beta_2}\right)U \left(\ket{0}_{\beta_1}^{\otimes a}\otimes I_{\beta_2}\right) \| \leq \epsilon.
\]    
\end{definition}
In particular, when $\epsilon=0$, we have 
$$U=\begin{pmatrix}
H & \cdot\\
\cdot &\cdot
\end{pmatrix}$$
and $$U\ket{0}^{\otimes a}\ket{\psi}=\ket{0}^{\otimes a}H\ket{\psi}+\ket{\perp},$$ where the reduced state in the first $a$ qubits of
$|\perp\rangle$ is orthogonal to $|0\rangle^{\otimes a}$.

Once we have the unitary $U_H$ which is a block encoding of the matrix $H$, then the QSVT framework allows us to construct a unitary $V$ such that 
\[
V\ket{0}^{\otimes a} \ket{\psi}= \ket{0}^{\otimes a} f(H)\ket{\psi}+\ket{\perp},
\] where $f(H)$ is some polynomial of $H$.

In this paper, we use the result in \cite[Appendix B.2]{dalzell2024shortcut}, which employs the minimax polynomial defined in \cite{lin2020optimal} to implement the projection of the kernel of a rectangular matrix. For completeness, we restate the result as follows. 

The degree-$2\ell$ minimax polynomial $F_{\Delta,\ell}(x)$ is defined as ~\cite{lin2020optimal}
\begin{equation}\label{eq:minimaxpoly}
    F_{\Delta, \ell}(x) = \frac{T_\ell\left(\frac{1+\Delta^2-2x^2}{1-\Delta^2}\right)}{T_\ell\left(\frac{1+\Delta^2}{1-\Delta^2}\right)}
\end{equation}
where $T_\ell$ is the $\ell$th Cheybshev polynomial of the first kind, given by 
\begin{equation}
    T_\ell(z) = \begin{cases}
        \cos(\ell \arccos(z)) & \text{if } |z|\leq 1 \\
        \cosh(\ell \arccosh(z)) & \text{if } z >1 \\
        (-1)^\ell \cosh(\ell \arccosh(z)) & \text{if } z < -1
    \end{cases}\,.
\end{equation}

\begin{lemma}\label{lem:KP_poly_properties}
    The polynomial $F_{\Delta,\ell}$ is guaranteed to satisfy 
\begin{enumerate}
    \item For all $x \in [-1,1]$, it holds that $|F_{\Delta,\ell}(x)| \leq 1$.
    \item For all $x \in [\Delta, 1]$, it holds that $|F_{\Delta,\ell}(x)| \leq T_\ell(\frac{1+\Delta^2}{1-\Delta^2})^{-1} \leq \eta$. 
    \item $F_{\Delta,\ell}(0) = 1$. 
\end{enumerate}
\end{lemma}

Leaving $\Delta=1/\kappa(H)$ and $\eta=\epsilon$, it suffices to choose \begin{equation} \label{eq:ell_KP}
    \ell = \left\lceil \frac{\arccosh(\eta^{-1})}{\arccosh\left(\frac{1+\Delta^2}{1-\Delta^2}\right)} \right\rceil \leq \left\lceil \frac{1}{2\Delta}\ln\left(\frac{2}{\eta}\right) \right\rceil.
\end{equation}

\begin{lemma}[Kernel Projection] \label{lem:kernelprojection} Suppose we have a $(1,a,0)$ block encoding $U_H$ of the matrix $H$, which has no singular value in the interval $(0,\Delta)$. Let $\ket{\psi_0}=\sqrt{\gamma}\ket{\theta}+\sqrt{1-\gamma}\ket{\theta^{\perp}}$, where $\ket{\theta} \in Null(H), \ket{\theta^{\perp}}\in Row(H)$ Then we can construct a unitary $V$ such that 
\[
V \ket{0}^{\otimes a} \ket{\psi_0} = \ket{0}^{\otimes a} F_{\Delta,\ell}(H)\ket{\psi_0}
\]
using $\Theta(\ell)=\Theta(\kappa(H)\log(1/\epsilon))$ queries of $U_H$ and $U_H^{+}$. 

In particular, when measuring the first register to obtain all $\ket{0}^{\otimes a}$, we will get a quantum state $\ket{\theta'}$ with
\[
\|\ket{\theta'}-\ket{\theta}\|_2 \leq \epsilon
\]  
\end{lemma}
\begin{proof}
    Given the $(1,a,0)$ block encoding $U_H$, the QSVT framework allows us to construct a unitary $V$ that is a $(1,a+1,0)$ block encoding of the polynomial $F_{\Delta,\ell}(H)$ using $\Theta(\ell)$ calls to $U_H$ and $U_H^{+}$ such that:
    \[
    V \ket{0}^{\otimes a+1} \ket{\psi_0} = \ket{0}^{\otimes a+1} F_{\Delta,\ell}(H) \ket{\psi_0} + \ket{\perp},
    \]
    where $\ket{\perp}$ is some orthogonal state outside the desired subspace. For more details of this construction, we refer readers to \cite[Appendix B.1]{dalzell2024shortcut}.

    When we measure the first register and obtain $\ket{0}^{\otimes a+1}$ (which happens with probability $\|F_{\Delta,\ell}(H) \ket{\psi_0}\|_2^2$), the resulting post-measurement state is:
    \[
    \ket{\theta'} = \frac{F_{\Delta,\ell}(H) \ket{\psi_0}}{\| F_{\Delta,\ell}(H) \ket{\psi_0} \|} = \frac{1}{\| F_{\Delta,\ell}(H) \ket{\psi_0} \|} \left( \sqrt{\gamma} \ket{\theta} + \sqrt{1 - \gamma} F_{\Delta,\ell}(H) \ket{\theta^\perp} \right)
    \]

    Since $\ket{\theta} \in \text{Null}(H)$ and $\ket{\theta^\perp} \in \text{Row}(H)$, and by applying the properties from \cref{lem:KP_poly_properties}, we have:
    \[
    F_{\Delta,\ell}(H) \ket{\theta} = \ket{\theta}, \quad
    F_{\Delta,\ell}(H) \ket{\theta^\perp} \in \text{Row}(H), \quad \| F_{\Delta,\ell}(H) \ket{\theta^\perp} \| \leq \eta
    \]
    
    Consequently, we obtain:
    \[
    \sqrt{\gamma} \leq \| F_{\Delta,\ell}(H) \ket{\psi_0} \| \leq \sqrt{\gamma} + \sqrt{1 - \gamma} \cdot \eta
    \]

    The overlap between $\ket{\theta'}$ and $\ket{\theta}$ can be bounded by:
    \[
    |\langle \theta' | \theta \rangle| \geq \frac{\sqrt{\gamma}}{\sqrt{\gamma} + \sqrt{1 - \gamma} \cdot \eta} = 1 - \frac{\eta \sqrt{1 - \gamma}}{\sqrt{\gamma} + \sqrt{1 - \gamma} \cdot \eta}
    \]

    Thus, the squared two-norm distance satisfies:
    \[
    \| \ket{\theta'} - \ket{\theta} \|_2^2 = 2 \left( 1 - |\langle \theta' | \theta \rangle| \right) \leq \frac{2 \eta \sqrt{1 - \gamma}}{\sqrt{\gamma} + \sqrt{1 - \gamma} \cdot \eta}
    \]

    Setting:
    \[
    \frac{2\eta \sqrt{1 - \gamma}}{\sqrt{\gamma} + \sqrt{1 - \gamma} \cdot \eta} = \epsilon^2
    \]
    we solve for $\eta$:
    \[
    2\eta \sqrt{1 - \gamma} = \epsilon^2 \left( \sqrt{\gamma} + \sqrt{1 - \gamma} \cdot \eta \right)
    \]
    
    Rearranging terms:
    \[
    \eta \left( 2\sqrt{1 - \gamma} - \epsilon^2 \sqrt{1 - \gamma} \right) = \epsilon^2 \sqrt{\gamma}
    \]
    
    Simplifying:
    \[
    \eta = \epsilon^2 \cdot \frac{\sqrt{\gamma}}{2\sqrt{1 - \gamma} \left( 1 - \epsilon^2 \right)}
    \]

    Finally, we conclude:
    \[
    \| \ket{\theta'} - \ket{\theta} \|_2 \leq \epsilon
    \]
\end{proof}

 \section{The relationship between $\kappa(H)$ and $\kappa(A)$}
In this section, we show that $\kappa(H)=\Theta(\kappa(A))$ under the assumptions that $\|A\|=  1$, $\|\vec{b}\|\leq 1$, $\|\xvec\|_2=1$. The condition number of a matrix $A$ is defined as the ratio between its maximum singular eigenvalue and its minimum nonzero singular eigenvalue $\frac{\sigma_{max}(A)}{\sigma_{min}^+(A)}$. To systematically investigate the relationship between $\kappa(H)$ and $\kappa(A)$, one way is to check the following matrix forms:

$$(H^{\top}H)_{N+1,N+1}=\begin{bmatrix}
    A^{\top}A & -A^{\top}\bvec\\
    -\bvec^{\top} A& \bvec^{\top}\bvec\\
\end{bmatrix},$$
and
$$(HH^{\top})_{M,M}= 
    AA^{\top}  
    + \bvec\bvec^{\top}.
 $$

\begin{theorem}\label{thm:kappaequ} Let $\|A\|=  1$, $\|\vec{b}\|\leq 1$, $\|\xvec\|_2=1$, and $H=[A, -\bvec]$; then we have
   \[
   \frac{1}{\sqrt{2}} \kappa(A) \leq  \kappa(H) \leq \sqrt{2}\cdot \kappa(A).
    \] 
\end{theorem}

\begin{proof} Since $\sigma_{max}(A) = 1$ as $\|A\|= 1$ and $\kappa(A)=\frac{\sigma_{max}(A)}{\sigma_{min}^+(A)}= \frac{1}{\sigma_{min}^+(A)}$, it is sufficient to prove that  $\sqrt{2}\sigma^+(A)\geq \sigma_{min}^{+}(H) \geq  \sigma_{min}^{+}(A)$ and $\sigma_{max}(H) \leq 2$.

\begin{itemize}
    \item $\sigma_{max}(H) \leq \sqrt{2}$:  The largest singular value of $H$ is equal to the square root of the largest eigenvalue of $HH^{\top}$.
    We know that 
    \[ (HH^{\top})_{M\times M}=AA^{\top}+\bvec\bvec^{\top}\]

    By Weyl's inequality, we have \[
    \lambda_{max}(HH^{\top})\leq \lambda_{max}(AA^{\top})+\lambda_{max}(\bvec\bvec^{\top})
    \]

    Therefore, we have \[
    \sigma_{max}(H)\leq \sqrt{ \lambda_{max}(HH^{\top})} \leq \sqrt{ \lambda_{max}(AA^{\top}) + \lambda_{max}(\bvec\bvec^{\top})} \leq \sqrt{2}.
    \]
    \item $\sigma_{min}^{+}(H) \geq  \sigma_{min}^{+}(A)$: We assume that the system $Ax = \bvec$ has a solution, which means that $ \bvec $ is in the column space of $A$. This implies that $r=rank(H)$ is equal to $r=rank(A)$ and they have the same number of nonzero singular eigenvalues.  This implies that $H$ has an extra zero singular value relative to $A$. 

    For $ H $, the Gram matrix is:

\[
H^\top H =
\begin{bmatrix}
A^\top A & -A^\top \bvec \\
-\bvec^\top A & \bvec^\top \bvec
\end{bmatrix} \in \mathbb{R}^{(N+1) \times (N+1)}.
\]

 By the Cauchy Interlacing Theorem,  if the nonzero eigenvalues of $ A^\top A $ are:

\[
\lambda_1(A^\top A) \geq \lambda_2(A^\top A) \geq \dots \geq \lambda_r(A^\top A),
\]

then the nonzero eigenvalues of $ H^\top H $, denoted as $ \lambda_1(H^\top H) \geq \lambda_2(H^\top H) \geq \dots \geq \lambda_{r}(H^\top H) $, satisfy:

\[
\lambda_1(H^\top H) \geq \lambda_1(A^\top A) \geq \lambda_2(H^\top H) \geq \lambda_2(A^\top A) \geq \dots \geq \lambda_r(H^\top H) \geq \lambda_r(A^\top A).
\]
Since singular values are the square roots of eigenvalues, we conclude:

\[
\sigma_{\min}^+(H) = \sqrt{\lambda_{\min}^+(H^\top H)} \geq \sqrt{\lambda_{\min}^+(A^\top A)} = \sigma_{\min}^+(A).
\]

\item $\sigma_{\min}^+(H)\leq \sqrt{2}\sigma_{\min}^+(A)$:
Let $\vec{u}_{\min}$ and $\vec{v}_{\min}$ be the left and right singular vectors of $A$
corresponding to $\sigma_{\min}^+(A)$, so that
\[
A v_{\min} =\sigma_{\min}^+(A) \vec{u}_{\min}, 
\qquad
A^\top \vec{u}_{\min} = \sigma_{\min}^+(A) \vec{v}_{\min}.
\]
Since $\vec{b} = A \vec{x}$ with $\|\xvec\|_2=1$, we have
\[
u_{\min}^\top \bvec 
= u_{\min}^\top A \xvec
= \sigma_{\min}^+(A) (\vec{v}_{\min}^\top \xvec),
\qquad\text{and hence}\qquad
|u_{\min}^\top \bvec| \leq \sigma_{\min}^+(A).
\]
Using the identity $HH^\top = AA^\top + \bvec \bvec^\top$, the Rayleigh quotient of $HH^\top$
at $\vec{u}_{\min}$ gives
\[
\vec{u}_{\min}^\top(HH^\top)\vec{u}_{\min}
= \vec{u}_{\min}^\top(AA^\top)\vec{u}_{\min} + (\vec{u}_{\min}^\top \bvec)^2
= (\sigma_{\min}^+(A))^2 + (\vec{u}_{\min}^\top \bvec)^2
\leq 2(\sigma_{\min}^+(A))^2.
\]
In addition, $\vec{u}_{\min}$ lies entirely within the subspace spanned by the eigenvectors of $HH^\top$
corresponding to its nonzero eigenvalues. Therefore, we have 

\[
\lambda_{\min}^+(HH^\top) \le  2(\sigma_{\min}^+(A))^2,
\qquad\text{and }\qquad
\sigma_{\min}^+(H) \leq
\sqrt{2}\,\sigma_{\min}^+(A).
\]

\end{itemize}  
\end{proof}

\end{document}

%% file: preamble.tex
\pdfoutput=1

 \usepackage{tikz}
\usetikzlibrary{trees}
\usepackage[english]{babel}
\usepackage[utf8]{inputenc}
\usepackage[T1]{fontenc}
\usepackage[pdftex, pdftitle={Article}, pdfauthor={Author}]{hyperref} 
\usepackage[margin=1in]{geometry} 

\usetikzlibrary {positioning}
\definecolor {processblue}{cmyk}{0.96,0,0,0}

\usepackage{comment}

\usepackage{fancyhdr}
\usepackage{braket}
\usepackage{amssymb,bbm}
\usepackage{amsmath}
\usepackage{mathtools}

\newcommand{\bmacmat}{H}
\newcommand{\bvec}{ \vec{b}}
\newcommand{\s}{s}  
\newcommand{\pvec}{\vec{p}} 

\usepackage{latexsym}
\usepackage{amsthm}
\usepackage[capitalise]{cleveref}

\usepackage{comment}
\usepackage{mathdots}
\usepackage{url}
\usepackage{algorithm}
\usepackage{algorithmic}
\usepackage{appendix}

\newtheorem{theorem}{Theorem}[section]
\newtheorem{lemma}[theorem]{Lemma}
\newtheorem*{lemma-non}{Lemma}
\newtheorem*{theorem-non}{Theorem}

\newtheorem{corollary}[theorem]{Corollary}
\newtheorem{definition}[theorem]{Definition}

\newtheorem{problem}[theorem]{Problem}
\newtheorem{proposition}[theorem]{Proposition}
\newtheorem{example}[theorem]{Example}

\theoremstyle{definition}
{

}

\mathchardef\mhyphen="2D

\newcommand{\arccosh}{\mathrm{arccosh}}

\def\({\left(}
\def\){\right)}

\usepackage{graphicx}
\definecolor{greenn}{rgb}{0,0.8,0.2}
\definecolor{bluue}{rgb}{0.3,0,0.7}

\usepackage{ifdraft}
\ifdraft{\newcommand{\authnote}[3]{{\color{#3}{\text{ \bf #1:}} #2}}}{\newcommand{\authnote}[3]{}}
 
 \newtheorem{remark}[theorem]{Remark}
